\documentclass{article}

\usepackage{arxiv}

\usepackage[utf8]{inputenc} 
\usepackage[T1]{fontenc}    
\usepackage{hyperref}       
\usepackage{url}            
\usepackage{booktabs}       
\usepackage{amsfonts}       
\usepackage{nicefrac}       
\usepackage{microtype}      
\usepackage{graphicx}
\usepackage{natbib}
\usepackage{doi}

\usepackage{ntheorem}
\theoremstyle{plain}

\newtheorem{lemma}{Lemma}[subsection]
\newtheorem{Proposition}{Proposition}
\newtheorem{PropositionSM}{Proposition}[subsection]
\newtheorem*{proof}{Proof}

\newtheorem{Example}{Example}
\newtheorem{Remark}{Remark}

\usepackage{bbm, mathtools}
\usepackage{graphicx, caption, subcaption}
\usepackage[dvipsnames]{xcolor}
\usepackage{array}
\usepackage{pdflscape}
\usepackage{rotating}
\usepackage{booktabs}

\usepackage{algorithm}
\usepackage{algpseudocode}
\usepackage{xcolor}
\usepackage{tikz}
\usetikzlibrary{decorations.pathreplacing,arrows.meta,calc, patterns}
\newcolumntype{C}[1]{>{\raggedright\arraybackslash\small}m{#1}}
\newcolumntype{T}[1]{>{\raggedright\arraybackslash\small}m{#1}}

\newcommand{\argmin}{\operatornamewithlimits{argmin}}

\providecommand{\E}{\mathbb{E}}
\providecommand{\Var}{\mathrm{Var}}
\providecommand{\Cov}{\mathrm{Cov}}
\providecommand{\sign}{\mathrm{sign}}
\providecommand{\uni}{\textcolor{blue}{Uni}}
\providecommand{\multi}{\textcolor{magenta}{Multi}}

\usepackage[normalem]{ulem}  

\title{A Note on the Folding Test of Unimodality: limitation and improved alternative}

\author{ \href{https://orcid.org/0009-0003-0790-3720}{\includegraphics[scale=0.06]{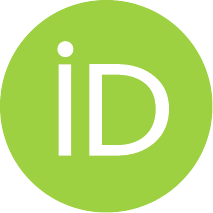}\hspace{1mm}Colombe Becquart} \thanks{Corresponding author: colombe.becquart@tse-fr.eu} \\
Toulouse School of Economics,\\ Université Toulouse Capitole  \\ France \\
Université de Toulouse \\
 France \\
	\And
	\href{https://orcid.org/0000-0002-6511-9091}{\includegraphics[scale=0.06]{orcid.pdf}\hspace{1mm}Aurore Archimbaud} 
        \\
	TBS Business School \\ France \\
	\And
    \href{https://orcid.org/0000-0001-8970-8061}{\includegraphics[scale=0.06]{orcid.pdf}\hspace{1mm}Anne M. Ruiz} \\
	Toulouse School of Economics,\\ Université Toulouse Capitole \\
 France \\
    \And
    \href{https://orcid.org/0000-0002-9974-299X}{\includegraphics[scale=0.06]{orcid.pdf}\hspace{1mm}Zaineb Smida} \\
	Univ Lyon, INSA Lyon, UJM, UCBL, ECL, \\CNRS UMR 5208, ICJ\\ France \\
}

\renewcommand{\shorttitle}{A Note on the FTU}

\hypersetup{
pdftitle={A Note on the Folding Test of Unimodality: limitation and improved alternative},
pdfsubject={62G10, 62H30},
pdfauthor={Colombe Becquart, Aurore Archimbaud, Anne M.~Ruiz, Zaineb Smida},
pdfkeywords={Dip test, Mixture distribution, Multimodality, Unimodality test},
}

\begin{document}
\maketitle

\begin{abstract}
	This note addresses a key limitation of the Folding Test of Unimodality (FTU). In specific univariate mixture settings, the folding-based criterion can systematically fail, misclassifying clearly multimodal distributions as unimodal.
    We fully characterize these failures for Dirac mixtures and extend the analysis to Gaussian mixtures. We then introduce a double-folding procedure that captures complementary information, leading to a new test, the Double Folding Test of Unimodality. It resolves the FTU failures and improves multimodality detection power in simulations.
\end{abstract}


\section{Introduction}

Determining whether a distribution is unimodal or multimodal is a fundamental problem in statistical analysis, particularly in the context of clustering, where the presence of multiple modes is often interpreted as evidence of underlying cluster structure. This connection has been widely exploited, both within clustering algorithms \citep{kalogeratos2012dip, shahbaba2014efficient} and as a preprocessing step \citep{adolfsson2019cluster, krause2005multimodal}.
In particular, the interplay between dimension reduction and unimodality provides valuable tools for assessing clusterability \citep{adolfsson2019cluster} and arises naturally in projection pursuit \citep{krause2005multimodal} and tandem clustering \citep{arabie_cluster_1994, alfons_tandem_2024}.

According to \citet{khintchine_unimodal_1938}, a real random variable is said to be unimodal, with mode $a$, if its cumulative distribution function is convex on $(- \infty, a)$ and concave on $(a, +\infty)$.
A standard reference for testing unimodality is Hartigan's dip test \citep{hartigan1985dip}, which is based on the cumulative distribution function. Another commonly used test is Silverman's test \citep{silverman_using_1981}, relying on kernel density estimation.

The Folding Test of Unimodality (FTU) proposed by \citet{siffer_are_2018} offers an appealing alternative: it is computationally efficient, easy to implement, and particularly well suited to streaming data.
Its rationale is that folding a multimodal distribution greatly reduces its variance, whereas the reduction is much smaller for a unimodal distribution.
In the univariate case, for a real random variable~$X$ with finite second moment, folding consists of applying the transformation $X \mapsto |X - s|$, where~$s$ is the folding pivot. The associated variance reduction is measured by the so-called folding ratio:
 \begin{equation}\label{eq:phi}
\phi^*(X) = \displaystyle \min_{s\in \mathbb{R}} \frac{\Var|X-s|}{\Var X}.
 \end{equation}

The smaller $\phi^*(X)$, the more effective the folding is at reducing the variance. The standardized folding ratio (SFR), called the folding statistic in \cite{siffer_are_2018}, compares this ratio to that of a reference unimodal distribution. Following \citet{hartigan1985dip}, \citet{siffer_are_2018} use a random variable $U$ following a uniform distribution $\mathcal{U}([l,u])$ on any bounded interval $[l, u]$ as the least concentrated unimodal case.
Then, the SFR is defined by
$$
\Phi^*(X) = \frac{\phi^*(X)}{\phi^*(U)}.
$$

By Definition 7.2 in \citet{siffer_new_2019}, if $\Phi^*(X) \geq 1$, the distribution is considered unimodal; otherwise, it is deemed multimodal.
The underlying heuristic of folding a distribution and measuring the resulting variance reduction is not equivalent to Khintchine’s definition. This motivated \citet{siffer_new_2019} to study the behavior of the SFR under well-known unimodal distributions. His results provide evidence that variance reduction reflects unimodality for Beta and Gamma distributions.
However, the limitation of this heuristic is illustrated by a counterexample repeatedly discussed in the literature \citep{siffer_new_2019, chasani_uu-test_2022, kolyvakis_multivariate_2023}: the balanced mixture of three symmetric Dirac masses at $-a$, $0$, and $a$. We refer to this case as the ``pathological'' example, together with the corresponding Gaussian mixture with small noise. According to Khintchine’s definition, the latter is clearly multimodal, whereas the folding-based criterion concludes unimodality. This example reveals a fundamental blind spot of folding ratios and highlights their limitation in capturing multimodality.

This paper addresses this limitation through a detailed analysis of the SFR. 
We make two main contributions. 
First, we characterize the limitation of the SFR beyond the pathological example, by identifying structural conditions under which mixture models are misclassified.
Second, we introduce a double-folding procedure that yields a revised unimodality criterion based on two complementary folding ratios, computed on the original and folded data, and an associated test, the Double Folding Test of Unimodality (DFTU), which resolves these failures.

Section~\ref{sec:folding} provides a deeper understanding of the folding mechanism through the study of Dirac and Gaussian mixtures. Section~\ref{sec:doublefolding} introduces the double-folding procedure and the DFTU, while Section~\ref{sec:simu} reports the simulation results. Section~\ref{sec:conclu} concludes. All proofs, technical details, tables from the simulation studies, and additional illustrations are collected in the Appendix.
Replication files are available from: \url{https://github.com/cbecquart/DFTU-Replication}.

\section{Revisiting the standardized folding ratio}
\label{sec:folding}

Let $X$ be a real random variable with finite fourth moment. An exact pivot (not necessarily unique), denoted $s^*$, is derived from expression \eqref{eq:phi} and defined as $s^* = \displaystyle\argmin_{s \in \mathbb{R}} \Var|X-s|$.
This minimization problem is not necessarily convex; a solution always exists but may not be unique \citep[see p.119 in][and Appendix~\ref{app:existence}]{siffer_are_2018}.
To simplify the problem, \citet{siffer_are_2018} introduced the approximate pivot, denoted $s^{**}$, given by $s^{**} = \displaystyle\argmin_{s \in \mathbb{R}} \Var[(X-s)^2].$
In the univariate case, $s^{**}$ exists and is unique \citep[see][Theorem 3.1]{siffer_are_2018}. Its associated folding ratio is $\phi^{**}(X) = \Var|X - s^{**}|/\Var X$. In this setting, both $\phi^*(U)$ and $\phi^{**}(U)$ equal $1/4$ (see  Appendix~\ref{app:phi_1/4}). Thus the two SFR associated to $s^*$ and $s^{**}$ are, respectively, $\Phi^*(X) = 4\,\Var|X - s^*|$ and $\Phi^{**}(X) = 4\,\Var|X - s^{**}|$. 

\subsection{Mixture of Dirac distributions}
\label{subsec:dirac}

In the discrete case, let us consider a $k$-Dirac mixture with proportions $\varepsilon_i>0$ and locations $\mu_i$ for $i\in\{1,\ldots, k\}$, $k>1$: 
\begin{equation}\label{model:mixdirac}
X \sim \sum_{i=1}^{k} \varepsilon_i \,\delta_{\mu_i}    
\end{equation}
where $\sum_{i=1}^k\varepsilon_i=1$.
Using the affine invariance property of $\Phi^*(X)$ and $\Phi^{**}(X)$ (see Appendix~\ref{app:aff_inv}), let us consider without loss of generality (w.l.o.g.) that $X$ is standardized ($\E[X] = 0$ and $ \Var[X]=1$) with ordered locations $\mu_1 <  \cdots  < \mu_k$ in model \eqref{model:mixdirac}. 

For Dirac mixtures, one expects $\Phi^{*}(X) < 1$ and $\Phi^{**}(X) < 1$, indicating multimodality. However, there is at least the aforementioned pathological example where this is not the case. This balanced Dirac mixture ($\mu_1=-a$, $\mu_2=0$, $\mu_3=a$) results in $s^{**}=0$ and $\Phi^{**}(X) > 1$, while $s^{*}=\pm a/4$ yields $\Phi^{*}(X) = 1$. In this case, the SFR are considered to “fail” since they are not lower than one despite clear multimodality. The key question is whether similar failures occur in broader mixture models and whether they can be characterized.

For a 2-Dirac mixture, $\Phi^{*}(X) = \Phi^{**}(X)= 0 <1$ for any choice of proportions and locations, meaning the SFR always correctly detect multimodality.
In the general case, Propositions \ref{th:exact} and \ref{th:approx} provide explicit expressions for $\Phi^*(X)$ and $\Phi^{**}(X)$ respectively. They also give the conditions under which the ratios exceed~1, thus falsely indicating unimodality.

Suppose that model \eqref{model:mixdirac} holds. For any $g\in \{1,\ldots,k-1\}$, let 
$\alpha_g = \sum_{i=1}^{g} \varepsilon_i \mu_i$ and 
$\eta_g = \sum_{i=1}^{g} \varepsilon_i$. Proofs of the following propositions are available in Appendices \ref{app:thexact} and \ref{app:thapprox}.

\begin{Proposition} \label{th:exact}
For the exact pivot approach, we have:
\begin{enumerate}
    \item $s^* =  \frac{\alpha_{g^*}}{2}\left(\frac{1}{\eta_{g^*}} - \frac{1}{1-\eta_{g^*}}\right),$ where $g^* = \displaystyle\argmin_{g\in \{1,\ldots,k-1\}} \min_{s\in [\mu_g, \mu_{g+1}]} \Var|X-s|$. 
    \item The associated SFR is $\Phi^*(X) = 4 \Big(1 - \frac{\alpha_{g^*}^2}{\eta_{g^*}(1-\eta_{g^*})} \Big)$.
    \item The equation \( \Phi^*(X) = 1 \) defines an ellipse in the \((\alpha_{g^*}, \eta_{g^*})\)-plane with center \((0, \tfrac{1}{2})\), semi-minor axis \(\tfrac{\sqrt{3}}{4}\), and semi-major axis \(\tfrac{1}{2}\).
    \item The SFR fails, i.e. \(\Phi^*(X) \geq 1\), if and only if
$\alpha_{g^*} \in \left[-\frac{\sqrt{3}}{2}\sqrt{\eta_{g^*}\bigl(1-\eta_{g^*}\bigr)},\, 0\right]$.

\end{enumerate}
\end{Proposition}

\begin{Proposition} \label{th:approx}
For the approximate pivot approach, we have: 

\begin{enumerate}
    \item $s^{**} = \frac{\gamma}{2}, \text{~with~} \gamma = \E[X^3]$. 
    \item The associated SFR is $\Phi^{**}(X) = 4 \Big(1 - 4 \alpha_{g^{**}}^2-2\alpha_{g^{**}}(1-2\eta_{g^{**}})\gamma +\eta_{g^{**}}(1-\eta_{g^{**}})\gamma^2 \Big)$, where $g^{**} = \sum^k_{i=1} \mathbbm{1}_{\{\mu_i \leq s^{**}\}}$.
    \item The equation $\Phi^{**}(X) = 1$ corresponds to a line in the $(\alpha_{g^{**}}, \eta_{g^{**}})$-plane, given by $4\alpha_{g^{**}} - (2 \eta_{g^{**}}-1)\,\gamma +\sqrt{\gamma^2+3} = 0$.
    \item The SFR fails, i.e. $\Phi^{**}(X) \geq 1$, if and only if $4\alpha_{g^{**}} - (2\eta_{g^{**}}-1)\,\gamma +\sqrt{\gamma^2+3} \geq 0$.
\end{enumerate}
\end{Proposition}

\begin{Remark} \label{rem:perf}
Since $\Phi^{*}(X) \leq \Phi^{**}(X)$, if the exact pivot approach fails ($\Phi^*(X) \geq 1$), then the approximate pivot also fails ($\Phi^{**}(X) \geq 1$).
\end{Remark}

Propositions \ref{th:exact} and \ref{th:approx} show that the SFR fail in certain cases. These cases are not necessarily characterized by balance, symmetry, or a three-group structure, as in the pathological example. Example \ref{ex:pivot} illustrates two failure cases that differ from the pathological example. It also highlights a situation where only the approximate pivot fails (case (a)). 

\begin{Example} \label{ex:pivot}
\begin{enumerate}
    \item[(a)] Let $X \sim 0.2 \cdot \delta_{-2} + 0.4 \cdot \delta_{0} + 0.4 \cdot \delta_{2}$. In this unbalanced case, only the approximate pivot fails: $\Phi^{**}(X) = 1.41>1$ while $\Phi^{*}(X) = 0.95 <1$.
    \item[(b)] Let $ X \sim 0.2 \cdot \delta_{-3} + 0.2 \cdot \delta_{-1.5} + 0.2 \cdot \delta_{2.5} + 0.2 \cdot \delta_{4} + 0.2 \cdot \delta_{11}$. In this asymmetric case, both pivots fail: $\Phi^{**}(X) = 1.38 > 1$ and $\Phi^{*}(X) = 1.07 > 1$.
\end{enumerate}

\end{Example}

\subsection{Mixture of Gaussian distributions}
    
In the continuous setting, let us consider a $k$-Gaussian mixture displaying a within-variance term, with proportions $\varepsilon_i>0$, means $\mu_i$, and variances $\sigma_i^2$, for $i \in \{1, \ldots, k\}$, $k>1$:
\begin{equation}\label{eq:mixtgauss}
    X \sim \sum_{i=1}^{k} \varepsilon_i \, \mathcal{N}(\mu_i, \sigma_i^2),
\end{equation}
where $\sum_{i=1}^{k} \varepsilon_i = 1$, and $\E[X] = 0$. 
In this setting, we cannot obtain an explicit expression for the exact and approximate pivots.
The SFR, can however, be derived from the expression:

\begin{equation}
\label{eq:gauss}
    \Phi(X, s) =  4 \Big[ \Var[X] + s^2 - \Big( \sum_{i=1}^{k} \varepsilon_i\, h(\mu_i, \sigma_i, s)\Big)^2\Big],
\end{equation}
where 
$h(\mu, \sigma, s) = (s-\mu)\left( 2 F\left( \frac{s-\mu}{\sigma} \right) - 1\right) + 2 \sigma f\left( \frac{s-\mu}{\sigma}\right),$
and where $f$ and $F$ denote, respectively, the probability density function and the cumulative distribution function of the standard normal distribution $\mathcal{N}(0,1)$.
Equation~\eqref{eq:gauss} is easily derived from the variance of the folded normal distribution \citep{leone1961folded}. 

We have $\Phi^*(X) = \Phi(X, s^*)$ and $\Phi^{**}(X) = \Phi(X, s^{**})$.
When all variances $\sigma_i$ are set to 0 for $i = \{1, \ldots, k\}$, the expressions $\Phi^{*}(X)$ and $\Phi^{**}(X)$ obtained from equation \eqref{eq:gauss} reduce to the Dirac case (Propositions \ref{th:exact} and \ref{th:approx}).

Note that one can standardize $X$ such that $\Var[X] = 1 + W$, where $ W = \sum_{i=1}^k \varepsilon_i \sigma_i^2$ represents the within-variance. Accordingly, the signal-to-noise ratio (SNR) is defined as $\mathrm{SNR} = 1/W$. Setting the between-variance to~1 ensures consistency with Subsection \ref{subsec:dirac}, which corresponds to the special case where $W = 0$.
Since theoretical computations beyond this point are challenging, we consider further the 2-Gaussian mixture case and highlight the impact of the within-variance term. The results differ from the Dirac case, where the SFR always detects multimodality for $k=2$, as shown in Example \ref{ex:gauss}.

\begin{Example}\label{ex:gauss}

Let $X \sim 0.3 \cdot \mathcal{N}(-2.8, \sigma^2) + 0.7 \cdot \mathcal{N}(1.2, \sigma^2)$. Fig.~\ref{fig:phi_s} shows $\Phi^*(X)$ as a function of $\sigma^2$, while Fig.~\ref{fig:pdf} and \ref{fig:cdf} illustrate the distribution for some values of $\sigma^2$. 
When $\sigma^2$ is small (high SNR), the distribution resembles a 2-Dirac mixture and $\Phi^*(X)$ is close to zero. As $\sigma^2$ increases, the two modes become less distinct (Fig.~\ref{fig:pdf}), the distribution looks more unimodal (Fig.~\ref{fig:cdf}) and $\Phi^*(X)$ rises (Fig.~\ref{fig:phi_s}).
In addition, when $\sigma^2$ exceeds $1.34$ (Fig.~\ref{fig:phi_s}), $\Phi^*(X)$ becomes greater than $1$, indicating the transition from multimodality to unimodality according to Definition 7.2 of unimodality in \cite{siffer_new_2019}. Note, however, that the distribution with $\sigma^2 = 1.34$ would still be considered multimodal under the definition of \cite{khintchine_unimodal_1938}, as seen from the multiple changes of convexity in Fig.~\ref{fig:cdf}.

\end{Example}

\begin{figure}[ht]
    \centering
    \begin{subfigure}[b]{0.9\textwidth}
        \centering
        \includegraphics[width=\textwidth]{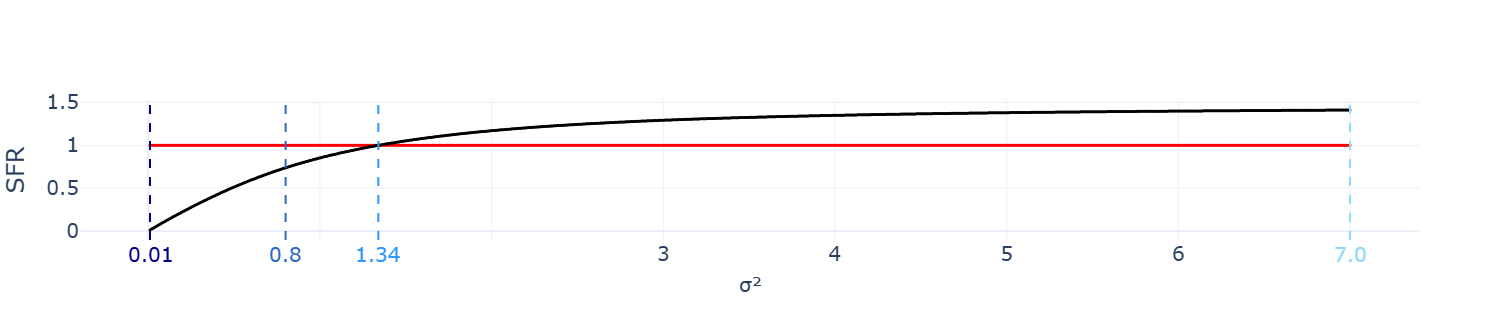}
        \caption{SFR $\Phi^*(X)$ (top)}
        \label{fig:phi_s}
        \vspace{5mm}
    \end{subfigure}
    \begin{subfigure}[b]{\textwidth}
        \centering
        \includegraphics[width=\textwidth]{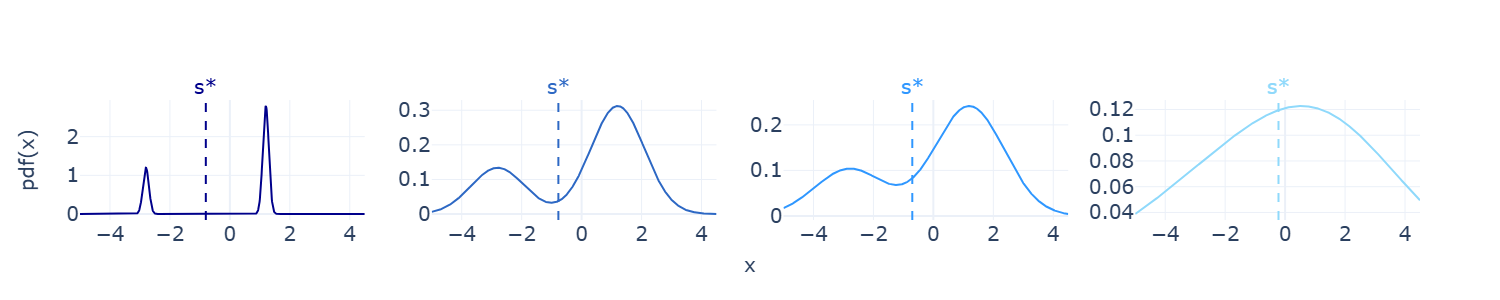}
        \caption{Probability density  functions}
        \label{fig:pdf}
        \vspace{5mm}
    \end{subfigure}
    \begin{subfigure}[b]{\textwidth}
        \centering
        \includegraphics[width=\textwidth]{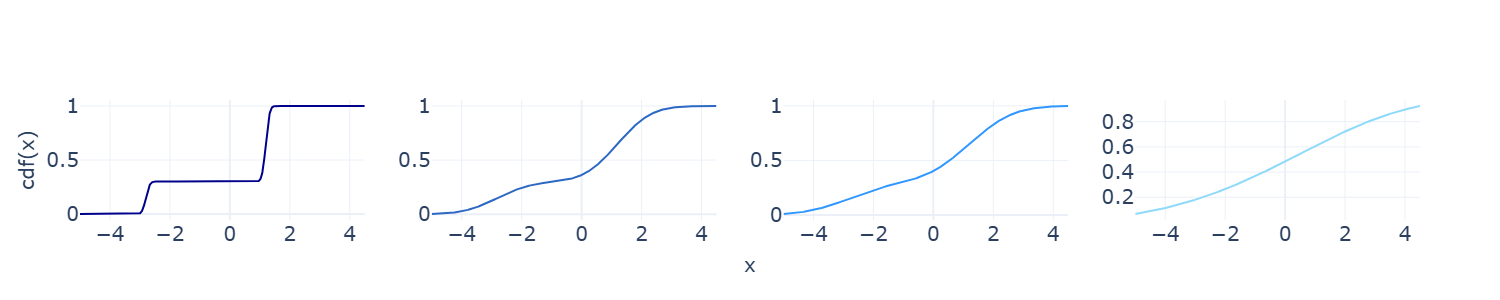}
        \caption{Cumulative distribution functions}
        \label{fig:cdf}
    \end{subfigure}

    \begin{subfigure}[b]{0.7\textwidth}
        \centering
        \includegraphics[width=\textwidth]{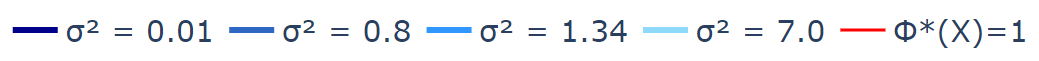}
    \end{subfigure}
    \caption{2-Gaussian mixture with $\varepsilon_1 = 0.3$, $\mu_1 = -2.8$, $\varepsilon_2 = 0.7$, $\mu_2 = 1.2$, and varying $\sigma^2$: (a) values of $\Phi^*(X)$ (black) as a function of $\sigma^2$, (b) probability density functions and associated exact pivots $s^*$, (c) cumulative distribution functions.}
    \label{fig:gaussian}
\end{figure}

The remarks of Example \ref{ex:gauss} are not specific to this particular example. For Gaussian mixtures \eqref{eq:mixtgauss} with low SNR that nearly satisfy Khintchine’s definition of unimodality, the SFR may conclude unimodality ($\Phi^*(X) \geq 1$). In that situation, the ratio cannot be considered to have failed.
In other cases, however, $\Phi^*(X) \geq 1$ reflects a genuine failure caused by the mixture structure, as in the Dirac case (see Appendix~\ref{app:gaussian} for details).

In the following section, we propose a revised unimodality criterion and an associated test, designed to correctly identify multimodality in the aforementioned failure cases.

\section{Enhancing the standardized folding ratio with a double-folding procedure}
\label{sec:doublefolding}

\subsection{The double-folding algorithm}

The limitation of the SFR in detecting multimodality is addressed through the development of an iterative algorithm. 
The first step consists of computing the exact SFR $\Phi_1^*(X)$ with associated pivot $s_1^*$. If $\Phi_1^*(X) < 1$, the procedure stops and the algorithm returns a multimodal conclusion. Otherwise, a second step is performed. The data are folded using the approximate pivot $s_1^{**}$, and the exact SFR $\Phi_2^*(X)$, with associated pivot $s_2^*$, is then computed on the folded data. The procedure concludes unimodality if $\Phi_2^*(X) \geq 1$, and multimodality otherwise.

The use of exact SFR in both steps is justified by Remark~\ref{rem:perf}.
The rationale for using the approximate pivot to transform the data is its ability to break the symmetry of the data when necessary. Indeed, folding the pathological example's distribution using an exact pivot yields another pathological mixture, with the gap between groups reduced by half (see Appendix~\ref{app:patho}).
The second step will thus also falsely conclude unimodality. In contrast, using the approximate pivot breaks the initial symmetry, enabling the detection of multimodality in the second step. An alternative would be to use an exact pivot shifted by a small delta value.

The second step is not always necessary. According to our empirical experience, the SFR do not falsely conclude multimodality; the only issue arises when they incorrectly conclude unimodality. If the second step were applied even when the first step ratio detects multimodality, it could override the correct detection. For instance, folding a 2-Dirac mixture may lead to a unimodal distribution. 

If the data are unimodal, both steps should correctly conclude unimodality ($\Phi_1^*(X) \geq 1$ and $\Phi_2^*(X) \geq 1$). 
If they are multimodal, either the first step detects it or, if not, we conjecture that the second step reveals it ($\Phi_1^*(X) < 1$ or $\Phi_2^*(X) < 1$).
In particular, the double-folding procedure yields the correct conclusion for the ``pathological'' example. When folding with the approximate pivot ($s_1^{**}=0$), the resulting distribution becomes a mixture of two Dirac masses at $0$ and $a$, with respective weights $1/3$ and $2/3$. This leads to $s_2^{*} = a/2$ and $\Phi_2^{*} = 0$, indicating clear multimodality.
The following subsection analyzes the general case of a 3-Dirac mixture to further clarify the motivation.

\subsection{Motivation through the 3-Dirac mixture case}

Let us consider the standardized model \eqref{model:mixdirac} with $k=3$ and ordered locations. Throughout this section, when the minimizer of $\Var|X-s|$ is not unique, we take $s^*$ to be the smallest minimizer.
We provide a new characterization of Proposition~\ref{th:exact}, in which the exact SFR failure conditions depend only on $\mu_1$, $\varepsilon_1$, and $\varepsilon_3$. It allows us to analyze the second step of the double-folding procedure when the first-step ratio fails. Then, based on numerical evidence, we argue that in this setting the double-folding correctly detects multimodality. This characterization is given in Proposition~\ref{prop:dftu3} (see Appendix~\ref{app:thdftu3}). 
Note that, by the affine invariance of the SFR, we may assume w.l.o.g. that $s^* \in [\mu_1,\mu_2]$.

\begin{Proposition} \label{prop:dftu3}
    \begin{enumerate}
        \item The existence of $s^* \in [\mu_1, \mu_2]$ is equivalent to $\mu_1 \leq U(\varepsilon_1, \varepsilon_3)$ with $ U(\varepsilon_1, \varepsilon_3) \coloneqq \max\{U_1, U_2\}$ where 
        
        $U_1 = - \frac{\sqrt{2}}{2} \sqrt{\frac{1 - \varepsilon_1}{\varepsilon_1} + \sqrt{\frac{1 - \varepsilon_1}{\varepsilon_1}} \sqrt{\frac{\varepsilon_3}{1 - \varepsilon_3}}}$ and 
        $U_2 = - \frac{2\varepsilon_2 + 4\varepsilon_1 \varepsilon_3}{\sqrt{4\varepsilon_1 \varepsilon_3(\varepsilon_2 + 4\varepsilon_1 \varepsilon_3)}} \sqrt{\frac{\varepsilon_3}{1 - \varepsilon_3}}$.
        \item When $s^* \in [\mu_1, \mu_2]$, $\Phi^*(X) \geq 1$ if and only if $\mu_1 \geq F(\varepsilon_1)$, where the failure bound is given by $
    F(\varepsilon_1) = - \frac{\sqrt{3}}{2} \sqrt{\frac{1 - \varepsilon_1} {\varepsilon_1}}$.
    \end{enumerate}
\end{Proposition}

According to the double-folding algorithm, whenever the first-step ratio detects multimodality, the procedure stops. 
When it is not the case (Proposition~\ref{prop:dftu3}), the distribution is folded using the approximate pivot, the locations are reordered and standardized, and the bounds in Proposition~\ref{prop:dftu3} are updated accordingly. We analyze these updated quantities via numerical optimization.
First, the results show that $s_2^*$ lies between the first and second folded locations, since the inequality in Proposition~\ref{prop:dftu3}\,(1.) is satisfied for any admissible parameter values. Consequently, the inequality in Proposition~\ref{prop:dftu3}\,(2.) can be used to assess whether the second-step ratio fails. Numerical optimization indicates that this inequality does not hold, supporting the view that the double-folding procedure does not fail in this model (see Appendix~\ref{app:3dirac} for details).

Motivated by this result, we introduce a new unimodality test based on the double-folding procedure. The test is described in the following subsection.

\subsection{The Double Folding Test of Unimodality}

In \citet{siffer_are_2018}, the SFR is used to test the null hypothesis $H_0:\,\Phi^*(X)\geq 1$ at significance level $\alpha$, in accordance with their definition of unimodality. The introduction of a double-folding criterion provides a novel approach to testing unimodality, leading to the Double Folding Test of Unimodality (DFTU), which tests:
$$ \mathcal{H}_0: \Phi^*_1(X) \geq 1 \text{ and } \Phi^*_2(X) \geq 1 \quad vs. \quad \mathcal{H}_1: \Phi^*_1(X) < 1 \text{ or }\Phi^*_2(X) < 1,
$$
where 
$\Phi^*_1(X) = 4\,\Var|X - s^*_1| / \Var X$ with $s^*_1$ an exact pivot for the original data $X$, and 
$\Phi^*_2(X) = 4\,\Var|\tilde{X} - s^*_2| / \Var \tilde{X}$ with $\tilde{X} = |X - s^{**}_1|$, $s^{**}_1$ being the approximate pivot for $X$, and $s^*_2$ an exact pivot for the folded data $\tilde{X}$.

Regarding the statistical significance of the DFTU, the challenge lies in the fact that $\mathcal{H}_0$ involves two folding statistics $\Phi^*_1(X)$ and $\Phi^*_2(X)$. 
The two folding statistics, computed sequentially, are associated with significance levels $\alpha_1$ and $\alpha_2$, and corresponding critical values $q_1$ and $q_2$. Significance levels are $\alpha_1 = \mathbb{P}(\Phi^*_1(X) < q_1 |  \mathcal{H}_0)$ and $\alpha_2 = \mathbb{P}(\Phi^*_2(X) < q_2 | \Phi^*_1(X) \geq q_1; \mathcal{H}_0)$. In this case, $\mathcal{H}_0$ is rejected if either $\Phi^*_1(X) < q_1$, or if $\Phi^*_1(X) \geq q_1$ and $\Phi^*_2(X) < q_2$. Consequently, the overall significance level of the test $\alpha$ is given by $\alpha = \alpha_1 + (1 - \alpha_1) \alpha_2$. Given predefined values of $\alpha_1$ and $\alpha_2$, the corresponding critical values $q_1$ and $q_2$ can be estimated via Monte Carlo simulation.
Following \citet{siffer_new_2019}, we generate observations from a uniform distribution. The critical values $q_1$ and $q_2$ are then taken as the empirical quantiles of order $\alpha_1$ and $\alpha_2$ of $\Phi^*_1(X)$ and $\Phi^*_2(X)$, respectively. For details about the implementation of the DFTU, see Algorithm~\ref{alg:dftu} in Appendix~\ref{app:algo}.

\section{Simulations}\label{sec:simu}

We want to compare the proposed DFTU with the FTU \citep{siffer_are_2018} based on both pivots, and with the dip test \citep{hartigan1985dip}. 
For each distribution, we generate 100 datasets of 1,000 observations and apply all four tests at a 5\% significance level. For the DFTU, $\alpha_1$ is set to 3\%, although similar results are obtained for other choices of $\alpha_1$. 
Table \ref{tab:simu} reports the majority decision for each test (see Table~\ref{tab:simuraw} in Appendix~\ref{app:simu} for detailed results). The first three rows illustrate cases in which all tests reach the correct conclusion.
The standard Gaussian distribution and the 2-Gaussian mixture with low SNR are both classified as unimodal, whereas the well-separated mixture of a Gaussian and a uniform distribution is identified as multimodal by all tests.
The last four rows present multimodal examples in which the FTU does not detect multimodality, yet the DFTU accurately identifies it.
Finally, the mixture $0.6 \cdot \mathcal{N}(0, 0.5) + 0.4 \cdot \mathcal{U}([1, 4])$ is noteworthy, as folding-based tests and the dip test lead to different conclusions. According to the definition of unimodality based on the cumulative distribution function \citep{khintchine_unimodal_1938}, this distribution can be considered unimodal (see Table~\ref{tab:distrib} in Appendix~\ref{app:simu}). However, from a clustering perspective, it may be preferable to conclude multimodality, which matches the DFTU's conclusion, as well as the FTU$^{*}$'s and the FTU$^{**}$'s. 

\begin{table}[ht]
\centering
\begin{tabular}{lcccc}
\hline
Distribution & FTU$^{*}$  & FTU$^{**}$ & DFTU & dip test \\
\hline
$\mathcal{N}(0, 1)$ &  \uni   &  \uni   & \uni     &  \uni  \\
$0.5 \cdot (\mathcal{N}_0 + \mathcal{N}_1)$ & \uni & \uni & \uni & \uni  \\
$0.6 \cdot \mathcal{N}_0 + 0.4 \cdot \mathcal{U}([4, 8])$ & \multi & \multi & \multi & \multi  \\
$0.6 \cdot \mathcal{N}_0 + 0.4 \cdot \mathcal{U}([1, 4])$ & \multi & \multi & \multi & \uni  \\
$1/3 \cdot \delta_{-2} + 1/3 \cdot \delta_{0} + 1/3 \cdot \delta_{2}$ & \uni & \uni & \multi & \multi  \\
$1/3 \cdot \mathcal{N}_{-2} + 1/3 \cdot \mathcal{N}_{0} + 1/3 \cdot \mathcal{N}_{2}$ & \uni & \uni & \multi & \multi  \\
$0.2 \cdot (\delta_{-3} + \delta_{-1.5} + \delta_{2.5} + \delta_{4} + \delta_{11})$ & \uni & \uni & \multi & \multi  \\
$0.2 \cdot (\mathcal{N}_{-3} + \mathcal{N}_{-1.5} + \mathcal{N}_{2.5} + \mathcal{N}_{4} + \mathcal{N}_{11})$ & \uni & \uni & \multi & \multi  \\
\hline
\end{tabular}
\caption{Simulation results. FTU$^*$ is the FTU with exact pivot, FTU$^{**}$ is the FTU with approximate pivot. $\mathcal{N}_a = \mathcal{N}(a, 0.5)$. Uni and Multi stand respectively for unimodal and multimodal.}
\label{tab:simu}
\end{table}

\section{Concluding remarks}
\label{sec:conclu}

In this paper, we investigate a key limitation of the FTU, namely its inability to detect multimodality in certain mixture distributions, which stems from the behavior of the SFR. By analyzing Dirac and Gaussian mixtures, we identify distributional structures under which the SFR fails to capture multimodality. 

To address this issue, we propose the DFTU, a two-step test that applies the FTU to the original data and, when unimodality is not rejected, to the folded data. Computations and simulations show that it successfully detects multimodality in cases where the original FTU fails. These results suggest that combining information across successive transformations, as done through double folding, can substantially improve the detection of multimodality.

Although an iterative procedure with more than two steps is conceivable, we found no examples where this was necessary. Future work could focus on generalizing the DFTU to the multivariate case, as initially intended for the FTU, or on replacing the variance with a more robust estimator such as $Q_n$ \citep{rousseeuw_alternatives_1993}, since the current test may be sensitive to outliers. Another promising direction  is the application of the proposed approach to ICS tandem clustering.  
While informative
directions are traditionally identified through departures from normality, deviations from
unimodality may be more directly connected to group separation, which can be exploited for component selection.
All replication files are available at: \url{https://github.com/cbecquart/DFTU-Replication}.

\begin{appendix}

\renewcommand{\thetable}{\Alph{section}.\arabic{table}}
\setcounter{table}{0}
\renewcommand{\thealgorithm}{\thesubsection.\arabic{algorithm}}
\setcounter{algorithm}{0}

\section{Proofs}\label{appA}

\subsection{\texorpdfstring{Proof of the existence of $s^*$ }{}}\label{app:existence}

This subsection proves the existence of $s^*$. The proof follows p.~119 in \cite{siffer_new_2019}, adapted to the univariate case.  

\begin{PropositionSM}\label{prop:varbounds}
    Let $X$ be a real random variable (r.v.) with finite second moment. Then the function $s \mapsto \Var|X - s|$ is continuous and bounded on $\mathbb{R}$. In particular, $\forall s \in \mathbb{R}, 0 \leq \Var|X - s| \leq \Var X$.
\end{PropositionSM}

\begin{proof}
    We have
    $$\Var|X - s| = \Var X + \E^2[X - s] - \E^2[|X - s|]$$
    Define \(v : \mathbb{R} \to \mathbb{R}\) by \(v(s) = \Var|X - s|\). Since it is a composition of continuous functions, $v$ is continuous. Moreover, \(v(s) \ge 0\) for all \(s \in \mathbb{R}\), as it is a variance. 
    
    By Jensen's inequality, $| \E[X - s]| \leq \E[|X - s|]$, which implies $\E^2[X - s] \leq \E^2[|X - s|]$ and $\Var|X - s| \leq \Var X$.

\end{proof}

\begin{PropositionSM}\label{prop:existence}
    Let $X$ be a real r.v. with finite second moment. Then the function $s \mapsto \Var|X - s|$ attains an absolute minimum on $\mathbb{R}$.
\end{PropositionSM}

\begin{proof}
    
    We have
    \begin{align*}
    v(s)
    &= \Var X + \E^2[X - s] - \E^2[|X - s|] \\
    &= \Var X + s^2 \left(1 - \frac{\E[X]}{s}\right)^2
       - s^2 \left( \E\left[\left|1 - \frac{X}{s}\right|\right] \right)^2 .
    \end{align*}
    
    As $s \to \pm\infty$, we use the second-order expansion
    \[
    \E\!\left[|1+\alpha X|\right]
    = 1 + \alpha \E[X] + o(\alpha^2)
    \quad \text{as } \alpha \to 0.
    \]
    which yields
    \begin{align*}
    v(s)
    &= \Var X
    + s^2 \left(1 - \frac{\E[X]}{s}\right)^2
    - s^2 \left(1 - \frac{\E[X]}{s} + o\!\left(\frac{1}{s^2}\right)\right)^2 \\
    &= \Var X + o(1),
    \quad \text{as } s \to \pm\infty.
    \end{align*}

    The function $v$ is continuous and bounded between $0$ and $\Var X$
(Proposition~\ref{prop:varbounds}). Let $\varepsilon > 0$. Then there exists
$s_0 > 0$ such that $\Var X - \varepsilon \leq v(s) \leq \Var X \forall |s| > s_0$. The restriction of $v$ to $[-s_0,s_0]$ is then continuous on a bounded interval, so it attains its bounds by the Extreme Value Theorem.

\end{proof}

\subsection{\texorpdfstring{Proof of the affine invariance of $\phi^*(X)$ and $\phi^{**}(X)$ }{}}

\begin{PropositionSM} \label{prop:pivot_equi}
    The pivots $s^*$ and $s^{**}$ are equivariant under affine transformations. The corresponding folding ratios $\phi^*(X)$ and $\phi^{**}(X)$ are invariant under affine transformations. 
\end{PropositionSM}

\begin{proof}
    Let us consider an affine transformation of the random variable $X$, $\tilde{X} = aX+b$, for real numbers $a,b$ such that $a\neq0$.
    \begin{itemize}
        \item[(i)] We have:
        \[
        \begin{aligned}
        s^* = \arg\min_{s \in \mathbb{R}} \Var |X - s|
        &\iff \Var |X - s^*| \le \Var |X - s|,
        \quad \forall s \in \mathbb{R}, \\[0.5em]
        &\iff \Var |aX + b - a s^* - b|
           \le \Var |aX + b - a s - b|, \\
        &\qquad \forall s \in \mathbb{R}.
        \end{aligned}
        \]
    
    \noindent Let $\tilde{s} = as+b$ and $\tilde{s}^* = as^*+b$.
    
    Then, $\Var|\tilde{X}-\tilde{s}^*| \le \Var|X-\tilde{s}|, \forall \tilde{s} \in \mathbb{R} \;\iff\; \tilde{s}^* = \displaystyle\argmin_{\tilde{s} \in \mathbb{R}} \Var|\tilde{X}-\tilde{s}|$.
    
    Similarly, for $s^{**}$, we obtain $\tilde{s}^{**}=a s^{**}+b=\displaystyle\argmin_{\tilde{s} \in \mathbb{R}} \Var(\tilde{X}-\tilde{s})^2$.
    \item[(ii)] We also have:
$$\phi^{*}(\tilde{X}) = \frac{\Var|\tilde{X} - \tilde{s}^{*}|}{\Var\tilde{X}} = \frac{\Var|X - s^{*}|}{\Var X} = \phi^{*}(X) \quad \text{and} \quad \phi^{**}(\tilde{X})=\phi^{**}(X).$$
    \end{itemize}

\end{proof}

\subsection{\texorpdfstring{Proof of $\phi^*(U) =\phi^{**}(U) =  1/4$ }{}}\label{app:phi_1/4}

\begin{PropositionSM}\label{prop:phi_u}
The folding ratios $\phi^*(U)$ and $\phi^{**}(U)$ are equal to $1/4$, where $U \sim \mathcal{U}([l, u])$ on any bounded interval $[l, u]$. 
\end{PropositionSM}

\begin{proof}

\noindent By the affine invariance of $\phi^*(X)$ and $\phi^{**}(X)$, it is sufficient to consider $U \sim \mathcal{U}([-R, R])$ with $R > 0$.\\
\noindent For $\phi^*(U)$: Let $U \sim \mathcal{U}\big([-R, R]\big)$ for $R \in \mathbb{R}^*_{+}$.
We know that $\E[U] = 0$ and $\Var[U] = \frac{R^2}{3}$.
$$\phi^*(U) = \frac{\Var |U - s^*|}{\Var\,U},$$
where $s^* = \displaystyle\argmin_{s \in [-R, R]} \Var |U - s|$. For $s \in [-R, R]$, $\Var |U - s| = \E\big[(U-s)^2\big] - \E^2\big[|U - s|\big]$ where  $\E\big[(U-s)^2\big] = s^2 + \frac{R^2}{3}$ and $\E\big[|U - s|\big] = \frac{1}{2R}(R^2 + s^2)$. We obtain:
$$\Var |U - s| = \frac{1}{12} \Big( R^2 + 6s^2 - \frac{3s^4}{R^2} \Big), \quad \phi^*(U) = \frac{1}{4R^2} \Big( R^2 + 6s^2 - \frac{3s^4}{R^2} \Big).$$

There are three solutions to
$\frac{\partial \Var |U - s|}{\partial s} = s(1 - \frac{s^2}{R^2}) = 0$: $s=0$ or $s=R$ or $s=-R$. 
The only solution for which $\frac{\partial^2 \Var |U - s| }{\partial s^2} > 0$ is $s=0$. Thus, $s^*=0$, and the associated folding ratio $\phi^*(U) = 1/4$.\\
\noindent For $\phi^{**}(U)$: the proof is available in \cite{siffer_new_2019}.

\end{proof}

\subsection{Proof of the affine invariance of \texorpdfstring{$\Phi^*(X)$}{} and \texorpdfstring{$\Phi^{**}(X)$}{}}\label{app:aff_inv}

\begin{PropositionSM} \label{prop:inv}
    The SFR $\Phi^*(X)$ and $\Phi^{**}(X)$ are invariant by affine transformations. 
\end{PropositionSM}

\begin{proof}
    Let us consider an affine transformation of the random variable $X$, $\tilde{X} = aX+b$, where $a, b \in \mathbb{R}$ and $a\neq0$. From Proposition~\ref{prop:phi_u}, we know that $\phi^*(U) =\phi^{**}(U) =  1/4$, where $U \sim \mathcal{U}\big([l, u]\big)$ on any bounded interval $[l, u]$. Using the affine invariance of $\phi^*(X)$ and $\phi^{**}(X)$ (Proposition~\ref{prop:pivot_equi}), we have 
    
    $$\Phi^{*}(\tilde{X}) = \frac{\phi^*(\tilde{X})}{\phi^*(U)} = \frac{\phi^*(X)}{\phi^*(U)} = \Phi^{*}(X) \quad \text{and similarly} \quad \Phi^{**}(\tilde{X})=\Phi^{**}(X).$$

\end{proof}

\subsection{Proof of Proposition \protect\ref{th:exact}}\label{app:thexact}

In this paragraph, we provide the proof of the explicit expression for $\Phi^*(X)$ for the $k$-Dirac mixture model (2), and establish the conditions under which the ratio leads to false indications of unimodality.  $X$ is standardized ($\E[X] = 0$ and $ \Var[X]=1$) with ordered locations $\mu_1 <  \cdots  < \mu_k$. Let us suppose that for any $g\in \{1,\ldots,k-1\}$, we define $\alpha_g = \sum_{i=1}^{g} \varepsilon_i \mu_i$ and $\eta_g = \sum_{i=1}^{g} \varepsilon_i$.

In order to prove Proposition \ref{th:exact}, we will need the following lemmas. 

\begin{lemma}\label{lem:var}
This lemma is used in the proof of Proposition~\ref{th:exact}.

The following results give explicit expressions for $\Var|X-s|$.
\begin{enumerate}
    \item $\Var|X-s| = 1 + s^2 - \E^2[|X-s|].$ 
    \item For all $g\in \{1,\ldots,k-1\}$ and $s\in [\mu_g, \mu_{g+1}]$, the following three equivalent expressions hold:
    \begin{eqnarray*}
        \Var|X-s| &=& 1 - 4\eta_g(1-\eta_g)\left(\frac{\alpha_g}{\eta_g} - s\right)\left(\frac{\alpha_g}{1-\eta_g} + s\right)\\
        &=& 1-4(\alpha_g-\eta_g s)(\alpha_g+(1-\eta_g)s)\\
        &=& 4\eta_g (1-\eta_g)s^2+4(2\eta_g-1)\alpha_g s + 1-4\alpha_g^2.\\
    \end{eqnarray*}
    \end{enumerate}
\end{lemma}

\begin{proof}
\begin{enumerate} 

\item The expression is obtained by expanding:
\begin{eqnarray*}
\Var|X - s| &=& \E[(X - s)^2] - \E^2\bigl[|X - s|\bigr] \\
            &=& \E[X^2] - 2s\,\E[X] + s^2 - \E^2\bigl[|X - s|\bigr].
\end{eqnarray*}
Since $\E[X] = 0$ and $\E[X^2] = 1$, we obtain the result.

\item For all $g\in \{1,\ldots,k-1\}$ and $s\in [\mu_g, \mu_{g+1}]$, results are obtained using Lemma~\ref{lem:var}(1.) and the following expression:
    \begin{eqnarray*}
    \E^2\big[|X-s|\big] &=&  \big(- 2\alpha_g + s(2\eta_g - 1)\big)^2\\
    &=& 4(\alpha_g - s\eta_g)\big(\alpha_g + s(1-\eta_g)\big) + s^2.
\end{eqnarray*}
\end{enumerate}

\end{proof}

\begin{lemma} \label{lem:parab}
    This lemma is used in the proof of Proposition~\ref{th:exact}.
    
    The function $s\mapsto \Var|X-s|$ is continuous and consists of piecewise parabolic segments defined on each interval $[\mu_g,\mu_{g+1}]$, for $g\in \{1,\ldots,k-1\}$.
\end{lemma}

\begin{proof}
Lemma \ref{lem:var} states that the function is piecewise parabolic and we just need to prove that it is continuous. To do so, we use the last expression of $\Var|X-s|$ from Lemma \ref{lem:var}(2.). 
To verify the continuity of the function $s\mapsto \Var|X-s|$ at each point $s=\mu_g$, we evaluate the function from the left-hand limit ($s=\mu_g^-$) and the right-hand limit  ($s=\mu_g^+$), and confirm that they are equal.
Using the fact that $\alpha_g = \alpha_{g-1} + \varepsilon_g \mu_g$, and $\eta_g = \eta_{g-1} + \varepsilon_g$, we obtain:
\begin{eqnarray*}
    \Var|X-\mu_g^+| - \Var|X-\mu_g^-| &=& 0.
\end{eqnarray*}

\end{proof}

\begin{lemma} \label{lem:sg}
This lemma is used in the proof of Proposition~\ref{th:exact}.

For each $g \in \{1,\ldots,k-1\}$, define
$s_g =  \frac{\alpha_g}{2}\left(\frac{1}{\eta_g} - \frac{1}{1-\eta_g}\right)$. Then, the minimizer of $\Var|X-s|$ over the interval \([\mu_g, \mu_{g+1}]\) is given by
$$  s^\dagger_g =  \argmin_{s \in [\mu_g, \mu_{g+1}]} \Var|X-s|=  s_g - \min(0, s_g - \mu_g) -\max(0, s_g - \mu_{g+1}),
$$
which is equivalent to
\[
s_g^\dagger =
\begin{cases}
\mu_g, & \text{if } s_g < \mu_g,\\
s_g, & \text{if } s_g \in [\mu_g, \mu_{g+1}],\\
\mu_{g+1}, & \text{if } s_g > \mu_{g+1}.
\end{cases}
\]

An exact pivot can be chosen as any solution of the minimization problem 
$$
s^* = \argmin_{s \in \mathbb{R}} \Var|X-s| = s^\dagger_{g^*}, \quad \text{where} \quad g^* = \argmin_{g \in \{1,\ldots,k-1\}} \Var|X - s^\dagger_g|.
$$
    
\end{lemma}

\begin{proof}
From Lemma~\ref{lem:var}(1.), we have

\begin{eqnarray*}
        \Var|X-s| &=& 1 + s^2 - \Big(\sum_{i=1}^k\varepsilon_i\, |\mu_i-s|\Big)^2,
\end{eqnarray*}

and its derivative satisfies
 \begin{eqnarray*}
        \frac{d}{ds} \Var|X-s| &=& 0 \;\iff\; -s =  \Big(\sum_{i=1}^k\varepsilon_i\, \sign(\mu_i-s)\Big) \Big(\sum_{i=1}^k\varepsilon_i\, |\mu_i-s|\Big).
    \end{eqnarray*}

    Since $s\in [\mu_1, \mu_k]$, let us consider $\mu_1 < \cdots < \mu_g < s < \mu_{g+1} < \cdots < \mu_k$, for $g\in \{1, \ldots, k-1\}$. Then, $\sign(\mu_i-s) = -1, \forall i \in \{1, \ldots, g\}$, and $\sign(\mu_i - s) = 1, \forall i \in \{g+1, \ldots, k\}$. It yields
    $$ s_g = \frac{\alpha_g}{2}\left(\frac{1}{\eta_g} - \frac{1}{1-\eta_g}\right).$$
    
    Lemma \ref{lem:parab} states that the function $s\mapsto \Var|X-s|$ is piecewise parabolic. On each interval $[\mu_g, \mu_{g+1}]$, $g\in \{1, \ldots, k-1\}$, the point where the derivative vanishes corresponds to a minimum.

    Moreover, $s^* \in [\mu_1, \mu_k]$. Otherwise, $\Var|X-s| \geq 1 \text{ for all } s \in [\mu_1, \mu_k]$, which contradicts Proposition~\ref{prop:varbounds}.
    Hence, 
    $$
    s^* = \argmin_{s \in \mathbb{R}} \Var|X-s| = s^\dagger_{g^*}, \quad \text{where} \quad g^* = \argmin_{g \in \{1,\ldots,k-1\}} \Var|X - s^\dagger_g|.$$
    
\end{proof}

\begin{lemma} \label{lem:incrs}
This lemma is used in the proof of Proposition~\ref{th:exact}.

For all $g \in \{ 1, \ldots, k-2 \}$, the sequence $(s_g)_{g}$ is strictly increasing, i.e., $s_g < s_{g+1}.$
\end{lemma}

\begin{proof}
    We have $s_g = \frac{\alpha_g}{2} \Big(\frac{1}{\eta_g} - \frac{1}{1-\eta_g} \Big)$ for $g \in \{ 1, \ldots, k-1 \}$. Using $\alpha_{g+1} = \alpha_g + \mu_{g+1} \varepsilon_{g+1}$ and $\eta_{g+1} = \eta_g + \varepsilon_{g+1}$, for $g \in  \{ 1, \ldots, k-2 \}$, we obtain
    $$s_{g+1} - s_g = \frac{\varepsilon_{g+1}}{2(1-\eta_g)(1-\eta_{g+1})} \Big[ -\alpha_g - \mu_{g+1}(1-\eta_g) \Big] + \frac{\varepsilon_{g+1}}{2\eta_g \eta_{g+1}} (-\alpha_g + \mu_{g+1} \eta_g), $$
    where both $\frac{\varepsilon_{g+1}}{2(1-\eta_g)(1-\eta_{g+1})} > 0$ and $\frac{\varepsilon_{g+1}}{2\eta_g \eta_{g+1}} > 0$. As the locations are ordered in model (2):
    \begin{itemize}
        \item $\mu_i \leq \mu_g < \mu_{g+1} \quad \forall i \leq g$. This implies $-\alpha_g + \mu_{g+1}\eta_g >0$.
        \item $\mu_i \geq \mu_{g+1} \quad \forall i > g$. As $\E(X) = 0$, we have $\sum_{i=g+1}^k \varepsilon_i \mu_i = - \alpha_g$. It results in $-\alpha_g - \mu_{g+1}(1-\eta_g) \geq 0$.
    \end{itemize}
    Thus, $s_{g+1} - s_g > 0$. 
    
\end{proof}

\begin{proof} Proof of Proposition \ref{th:exact}

\begin{enumerate}
    \item We want to prove that $s^* = s_{g^*},$ where $g^* = \displaystyle\argmin_{g\in \{1,\ldots,k-1\}} \min_{s\in [\mu_g, \mu_{g+1}]} \Var|X-s|$ and $ s_g = \frac{\alpha_g}{2}\left(\frac{1}{\eta_g} - \frac{1}{1-\eta_g}\right)$.
    Let us prove the result by contradiction. By Lemma~\ref{lem:sg}, if $s^* \neq s_{g^*}$, then $s^* = \mu_{g+1}$ where $g\in \{1, \ldots, k-2\}$. In this case we would have $s^\dagger_g = \mu_{g+1}$ meaning that $s_g \geq \mu_{g+1}$, and $s^\dagger_{g+1} = \mu_{g+1}$ meaning that $s_{g+1} \leq \mu_{g+1}$. It implies that $s_{g+1} \leq s_{g}$ which contradicts Lemma~\ref{lem:incrs}.
    It results that $s^* = s_{g^*}$.
    \item The expression of $\Phi^*(X)$ is then obtained by substituting $s^*$ from the previous result into $\Phi^*(X) = 4 \, \Var|X - s^*|$. 
    \item Solving the equation $\Phi^*(X) = 1$ yields
    $\frac{\alpha_{g^*}^2}{(\sqrt{3}/4)^2} + \frac{\bigl(\eta_{g^*} - \frac{1}{2}\bigr)^2}{(1/2)^2} = 1$,
    which is the canonical equation of an ellipse in the variables $(\alpha_{g^*}, \eta_{g^*})$, centered at $(0, 1/2)$, with semi-minor axis $\sqrt{3}/4$ and semi-major axis $1/2$.
    \item Solving $\Phi^*(X) \geq1$, we obtain
    $$ \frac{\alpha_{g^*}^2}{(\sqrt{3}/4)^2} +   \frac{(\eta_{g^*}-\frac{1}{2})^2}{(1/2)^2} \leq 1 \;\iff\; \alpha_{g^*}^2 \leq \frac{3}{4} \eta_{g^*}(1-\eta_{g^*}).$$
    Since $\alpha_{g^*} \le 0$, this reduces to
    $-\frac{\sqrt{3}}{2}\sqrt{\eta_{g^*}(1-\eta_{g^*})} \le \alpha_{g^*} \le 0$ which proves the result.
    
\end{enumerate}
\end{proof}

\subsection{Proof of Proposition \protect\ref{th:approx}}\label{app:thapprox}
Here, we provide the proof of the explicit expression for $\Phi^{**}(X)$ for the $k$-Dirac mixture model, and the conditions under which the ratio falsely indicates unimodality.

\begin{proof} Proof of Proposition \ref{th:approx}
    \begin{enumerate}
        \item From \cite{siffer_new_2019} (Theorem 6.1), and using that $\E[X] = 0$ and $Var[X] = 1$, we obtain
        $$s^{**} = \frac{\Cov(X, X^2)}{2\Var X} = \frac{\E(X^3) - \E(X)\E(X^2)}{2\Var X} = \frac{\gamma}{2}, \text{~with~} \gamma = \E[X^3].$$

        \item The expression of $\Phi^{**}(X)$ is then obtained by substituting $s^{**}$ from the previous result into $\Phi^{**}(X) = 4 \, \Var|X - s^{**}|$. 

        \item Solving $\Phi^{**}(X) = 1$, we find
        $$\left(4\alpha_{g^{**}} - (2\eta_{g^{**}}-1)\gamma -\sqrt{\gamma^2+3}\right) \left(4\alpha_{g^{**}} - (2\eta_{g^{**}}-1)\gamma +\sqrt{\gamma^2+3}\right) = 0.$$
        Since $0 < \eta_{g^{**}} < 1$,  $-(2\eta_{g^{**}}-1)\gamma \le |\gamma|$. Since $\alpha_{g^{**}} < 0$ and $\sqrt{\gamma^2+3} > |\gamma|$, $4\alpha_{g^{**}} - (2\eta_{g^{**}}-1)\gamma -\sqrt{\gamma^2+3} < 0$. Thus, there is a unique solution, which defines the line
        $$4\alpha_{g^{**}} - (2\eta_{g^{**}}-1)\gamma +\sqrt{\gamma^2+3} =0.$$
        
        \item Solving $\Phi^{**}(X) \ge 1$, we obtain $4\alpha_{g^{**}} - (2\eta_{g^{**}}-1)\gamma +\sqrt{\gamma^2+3} \geq 0$.
        
    \end{enumerate}

\end{proof}

\subsection{Proof of Proposition \protect\ref{prop:dftu3}}\label{app:thdftu3}

We consider a standardized random variable $X$ following a 3-Dirac mixture with means $\mu_1$, $\mu_2$ and $\mu_3$, sorted in ascending order. When the minimizer of $\Var|X-s|$ is not unique, we take $s^*$ to be the smallest minimizer.

\begin{proof} Proof of Proposition \ref{prop:dftu3}

From the constraint \(\sum_{i=1}^3 \varepsilon_i = 1\), we deduce $\varepsilon_2 = 1 - \varepsilon_1 - \varepsilon_3$.
Using \(\mathbb{E}[X] = 0\), \(\mathrm{Var}[X] = 1\), and the ordering \(\mu_2 < \mu_3\), we obtain 
\[
\begin{aligned}
\mu_2 &= -\frac{\varepsilon_1}{1-\varepsilon_1}\,\mu_1 
        - \frac{1}{1-\varepsilon_1}
        \sqrt{\frac{\varepsilon_3}{1 - \varepsilon_1 - \varepsilon_3}}
        \,\sqrt{1-\varepsilon_1(1+\mu_1^2)} \\
\mu_3 &= -\frac{1}{\varepsilon_3}\,(\varepsilon_1 \mu_1 + \varepsilon_2 \mu_2).
\end{aligned}
\]

From the conditions \(1-\varepsilon_1(1+\mu_1^2) > 0\) and \(\mu_2 > \mu_1\), together with the fact that \(\mu_1 < 0\) in the model, we deduce

\begin{equation}\label{eq:mu1_interval}
\mu_1 \in
\left[
-\sqrt{\frac{1 - \varepsilon_1}{\varepsilon_1}},
\;
-\sqrt{\frac{\varepsilon_3}{1 - \varepsilon_3}}
\right].
\end{equation}

\begin{enumerate}
    \item We seek conditions on $\mu_1$ such that $s^* \in [\mu_1,\mu_2]$.
    By definition,
    \[
    s^* \in [\mu_1,\mu_2]
    \;\iff\;
    \Var[X - s_1^\dagger] \leq \Var[X - s_2^\dagger].
    \]
    By Proposition~\ref{th:exact}.1, this is equivalent to
    \[
    \Var[X - s_1] \leq \Var[X - s_2^\dagger],
    \]
    where, by Lemma~\ref{lem:sg},
    \[
    s_2^\dagger =
    \begin{cases}
    \mu_2, & \text{if } s_2 < \mu_2, \\
    s_2,   & \text{if } s_2 \geq \mu_2.
    \end{cases}
    \]
    
    Consequently,
    \[
    s^* \in [\mu_1,\mu_2]
    \;\iff\;
    \begin{cases}
    \Var[X - s_1] \leq \Var[X - \mu_2], & \text{if } s_2 < \mu_2, \\
    \Var[X - s_1] \leq \Var[X - s_2],   & \text{if } s_2 \geq \mu_2.
    \end{cases}
    \]

    Combining the cases $s_2 < \mu_2$ (see Lemma~\ref{lem:s2lower} below) and $s_2 \geq \mu_2$ (see Lemma~\ref{lem:s2greater} below), we conclude that
    \[
    s^* \in [\mu_1,\mu_2]
    \;\iff\;
    \mu_1 \leq \max\{U_1, U_2\}.
    \]
    \item When $s^* \in [\mu_1, \mu_2]$, we have $g^* = 1$.  
Applying Lemma~\ref{th:exact}.4 with $g^* = 1$, we have $\Phi^*(X) \geq 1$ if and only if
\[
\mu_1 \geq - \frac{\sqrt{3}}{2} \sqrt{\frac{1 - \varepsilon_1}{\varepsilon_1}} \geq - \sqrt{\frac{1 - \varepsilon_1}{\varepsilon_1}},
\]
which directly gives the expression of $F(\varepsilon_1)$.
\end{enumerate}

\end{proof}

\begin{lemma} \label{lem:s2lower}
This lemma is used in the proof of Proposition~\ref{prop:dftu3}.
\[
s^* \in [\mu_1,\mu_2] \;\text{and}\;s_2 < \mu_2
\;\iff\;
\mu_1 < U_1,
\]
where
\[
U_1
=
-\,\frac{2\varepsilon_2 + 4\varepsilon_1 \varepsilon_3}
{\sqrt{4\varepsilon_1 \varepsilon_3(\varepsilon_2 + 4\varepsilon_1 \varepsilon_3)}}
\sqrt{\frac{\varepsilon_3}{1 - \varepsilon_3}}.
\]
\end{lemma}

\begin{proof}

When $s^* \in [\mu_1,\mu_2]$, \(s_1\) minimizes \(\Var[X - s]\) over \(s \in [\mu_1,\mu_2]\), so the inequality
\[
\Var[X - s_1] \leq \Var[X - \mu_2]
\]
is always satisfied. It therefore remains to characterize the conditions under which
\(s_2 < \mu_2\). From Lemma~\ref{lem:sg},
\[
s_2 = -\dfrac{\mu_3}{2}\,\dfrac{2\varepsilon_3 - 1}{1 - \varepsilon_3}.
\]
Expressing $\mu_1$ and $\mu_2$ as functions of $\mu_3$ results in
\[
\mu_2
= -\dfrac{\varepsilon_3}{1 - \varepsilon_3}\mu_3
+ \dfrac{1}{1 - \varepsilon_3}
\sqrt{\dfrac{\varepsilon_1}{\varepsilon_2}}
\sqrt{1 - \varepsilon_3(1 + \mu_3^2)},
\]
\[
\mu_1
= -\dfrac{1}{1 - \varepsilon_3}
\Bigl(
\varepsilon_3 \mu_3
+ \sqrt{\dfrac{\varepsilon_2}{\varepsilon_1}}
\sqrt{1 - \varepsilon_3(1 + \mu_3^2)}
\Bigr).
\]
Imposing $\mu_2 < \mu_3$ and $1 - \varepsilon_3(1 + \mu_3^2) \geq 0$ yields
\begin{equation}\label{eq:mu3_interval}
\mu_3 \in
\left[
\sqrt{\frac{\varepsilon_1}{1 - \varepsilon_1}},
\;
\sqrt{\frac{1 - \varepsilon_3}{\varepsilon_3}}
\right].
\end{equation}
Solving $s_2 < \mu_2$, we obtain
\[
\mu_3
< 2 \sqrt{\dfrac{\varepsilon_1 (1 - \varepsilon_3)}
{\varepsilon_2 + 4 \varepsilon_1 \varepsilon_3}}.
\]
and the upper bound lies within the interval from \eqref{eq:mu3_interval}.
On this domain, $\mu_1$ is nondecreasing in $\mu_3$, so that
$s_2 < \mu_2
\;\iff\;
\mu_1 < U_1$.

\end{proof}

\begin{lemma} \label{lem:s2greater}
This lemma is used in the proof of Proposition~\ref{prop:dftu3}.
\[
s^* \in [\mu_1,\mu_2]
\;\text{and}\;
s_2 \geq \mu_2
\;\iff\;
\mu_1 \in [U_1, U_2],
\]
where
\[
U_2
=
- \frac{\sqrt{2}}{2}
\sqrt{
\frac{1 - \varepsilon_1}{\varepsilon_1}
+ \sqrt{\frac{1 - \varepsilon_1}{\varepsilon_1}}
\sqrt{\frac{\varepsilon_3}{1 - \varepsilon_3}}
}.
\]
\end{lemma}

\begin{proof}

Following the line of reasoning of the proof of Lemma~\ref{lem:s2lower},
\[
s_2 \geq \mu_2
\;\iff\;
\mu_1 \geq U_1.
\]
Moreover, solving the inequality
\[
\Var[X - s_1] \leq \Var[X - s_2]
\]
using the first variance expression from Lemma~\ref{lem:var}(2.) together with the explicit expressions of $s_1$ and $s_2$ from Lemma~\ref{lem:sg} yields
\[
\mu_1 \leq U_2.
\]
$U_2$ lies within the interval from \eqref{eq:mu1_interval}.
Therefore,
\[
s^* \in [\mu_1,\mu_2]
\;\text{and}\;
s_2 \geq \mu_2
\;\iff\;
\mu_1 \in [U_1, U_2].
\]

\end{proof}

\section{Additional Gaussian mixture illustrations}\label{app:gaussian}

In this section, we examine several 3-Gaussian mixtures and give supplementary figures, extending the analysis beyond Section 2.2 of the main paper. This additional study aims to confirm that the Gaussian mixture model reproduces the results obtained for the Dirac mixture model in cases where the SNR is sufficiently large.

Fig.~\ref{fig:gaussian_details} illustrates a 3-Gaussian mixture 
$X \sim \varepsilon_1 \cdot \mathcal{N}(\mu_1, \sigma^2) + \varepsilon_2 \cdot \mathcal{N}(\mu_2, \sigma^2) + \varepsilon_3 \cdot \mathcal{N}(\mu_3, \sigma^2)$, 
with $\sigma^2 \in \{ 0.0225, 0.0625, 0.25 \}$. The three configurations are as follows:
\begin{itemize}
    \item In Fig.~\ref{fig:gd1}: $(\varepsilon_1, \varepsilon_2, \varepsilon_3) = (1/3, 1/3, 1/3)$ and $(\mu_1, \mu_2, \mu_3) = (-2, 0, 2)$. It is an example where both $\Phi^*(X)$ and $\Phi^{**}(X)$ fail in the Dirac case ($\sigma^2=0$). 
    \item In Fig.~\ref{fig:gd2}: $(\varepsilon_1, \varepsilon_2, \varepsilon_3) = (0.2, 0.4, 0.4)$ and $(\mu_1, \mu_2, \mu_3) = (-2, 0, 2)$. It corresponds to Example 1.a in the main paper, where $\Phi^{**}(X)$ fails but $\Phi^{*}(X)$ does not. 
    \item In Fig.~\ref{fig:gd3}: $(\varepsilon_1, \varepsilon_2, \varepsilon_3) = (0.5, 0.2, 0.3)$ and $(\mu_1, \mu_2, \mu_3) = (-2, 0, 2.5)$. It illustrates a case where neither pivot fail in the Dirac case.
\end{itemize}

When both pivots fail (Fig.~\ref{fig:gd1}) or neither fails (Fig.~\ref{fig:gd3}) in the Dirac case, increasing $\sigma^2$ does not alter the conclusions of $\Phi^*(X)$ and $\Phi^{**}(X)$. 
When only the approximate pivot fails (Fig.~\ref{fig:gd2}), a sufficiently small $\sigma^2$ reproduces the Dirac results ($\Phi^*(X) < 1$ and $\Phi^{**}(X) > 1$).  
However, as $\sigma^2$ increases, $\Phi^*(X)$ exceeds one, leading to an incorrect unimodal conclusion.

\begin{figure}[htbp]
    \centering
    \begin{subfigure}[b]{\textwidth}
        \centering
        \includegraphics[width=\textwidth]{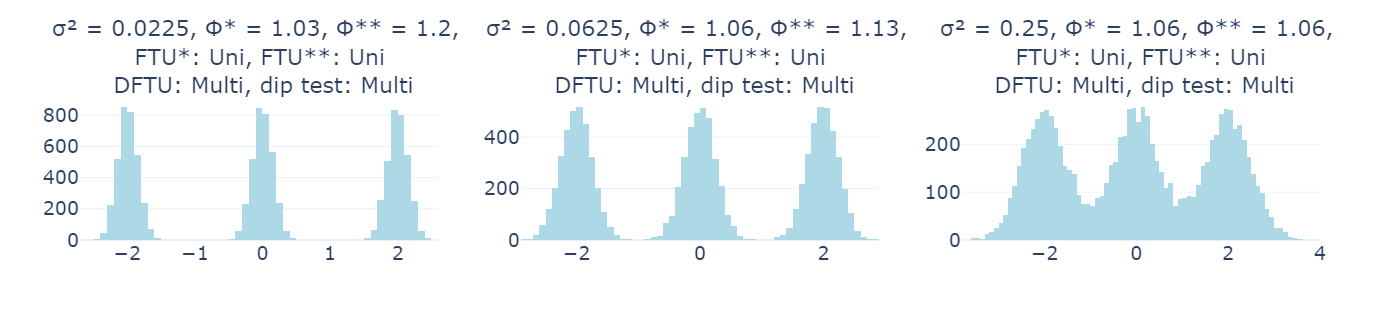}
        \caption{Both pivots fail in the Dirac case}
        \label{fig:gd1}
        \vspace{5mm}
    \end{subfigure}
    \begin{subfigure}[b]{\textwidth}
        \centering
        \includegraphics[width=\textwidth]{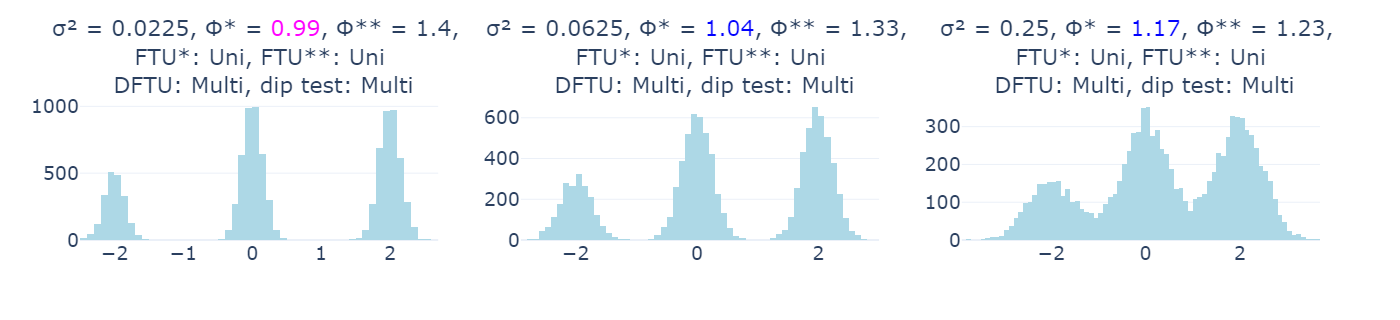}
        \caption{Only the approximate pivot fails in the Dirac case}
        \label{fig:gd2}
        \vspace{5mm}
    \end{subfigure}
    \begin{subfigure}[b]{\textwidth}
        \centering
        \includegraphics[width=\textwidth]{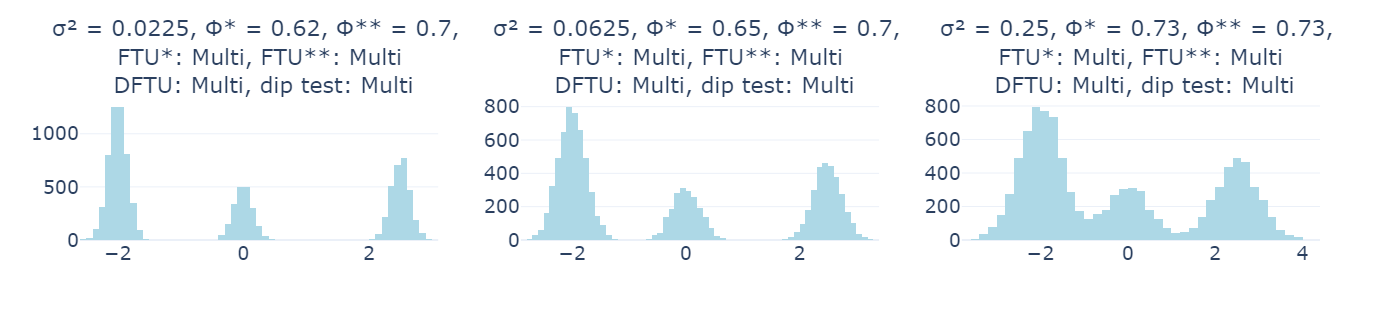}
        \caption{Neither pivot fails in the Dirac case}
        \label{fig:gd3}
    \end{subfigure}
    \caption{3-Gaussian mixtures under different configurations, together with the corresponding exact and approximate SFR ($\Phi^*$ and $\Phi^{**}$) and the results of unimodality tests. ``Uni” and ``Multi” denote unimodal and multimodal, respectively. Each column corresponds to a value of the mixture within variance $\sigma^2$. Each row corresponds to an example obtained in the Dirac case ($\sigma^2=0$): (a) both pivots fail, (b) only the approximate pivot fails, and (c) neither pivot fails.
    }
    \label{fig:gaussian_details}
\end{figure}

In all illustrated cases, both the DFTU and the dip test correctly detect multimodality, even when FTU$^*$ and FTU$^{**}$ fail to do so (Figs.~\ref{fig:gd1} and~\ref{fig:gd2}). These findings are consistent with the simulation results presented in Section 4 of the main paper.

\section{Details about the double-folding procedure}

\subsection{Application to the ``pathological'' example}\label{app:patho}

This section analyzes the application of the double-folding procedure to the ``pathological'' example, explains the use of the approximate pivot in the second step, and shows that the procedure correctly detects multimodality, as noted in Section 3.1 of the main paper.

The model is the following balanced and symmetric 3-Dirac mixture: 
$$ X \sim \tfrac{1}{3}\,\delta_{-a} + \tfrac{1}{3}\,\delta_{0} + \tfrac{1}{3}\,\delta_{a}. $$

The first step of the procedure corresponds to Example 8.1 from \cite{siffer_new_2019} with the exact pivot approach, yielding
$s_1^* = \pm \frac{a}{4}$ and $\Phi_1^*(X) = 1$. 
The first-step approximate pivot and associated SFR are
$s_1^{**} = 0$ and $\Phi_1^{**}(X) = \frac{4}{3} > 1$, as in Example 8.1.

Folding with an exact pivot yields the following distribution:
\[
\tilde{X}^* \sim \frac{1}{3} \delta_{\frac{1}{4}a} + \frac{1}{3} \delta_{\frac{3}{4}a} + \frac{1}{3} \delta_{\frac{5}{4}a}.
\] 
Since $\tilde{X}^* = \frac{1}{2} X + \frac{3}{4} a$ and the SFR are affine invariant, the unimodal conclusion of the first step persists in the second step.

However, folding with the approximate pivot, as done in the double-folding procedure, yields
\[
\tilde{X}^{**} \sim \frac{1}{3} \delta_0 + \frac{2}{3} \delta_a.
\] 
The folded distribution is no longer balanced or symmetric, leading to
$s_2^* = s_2^{**} = \frac{1}{4}$ and 
$\Phi_2^*(X) = \Phi_2^{**}(X) = 0$, 
which indicates multimodality.

\subsection{Details about the 3-Dirac mixture analysis}\label{app:3dirac}

In Section 3.2 of the main paper, we apply the double-folding procedure to a 3-Dirac mixture. This paragraph details the analysis, in particular the optimization procedure involved.

We consider a standardized random variable $X$ following a 3-Dirac mixture with locations $\mu_1 < \mu_2 < \mu_3$. In this case, when the minimizer of $\Var|X-s|$ is not unique, we take $s^*$ to be the smallest minimizer.
First, we analyze the first-step SFR. If the multimodality of the mixture is detected at this step, the procedure stops and returns the correct conclusion. We are therefore interested in the case where the first-step SFR fails, which can be fully characterized in terms of $\mu_1$, $\varepsilon_1$, and $\varepsilon_3$. 
According to Proposition~\ref{prop:dftu3}.1, the existence of $s_1^* \in [\mu_1, \mu_2]$ is equivalent to 
\[
\mu_1 \leq U(\varepsilon_1, \varepsilon_3), \quad
U(\varepsilon_1, \varepsilon_3) \coloneqq \max\{U_1, U_2\}
\] 
where 
$$U_1 = - \frac{\sqrt{2}}{2} \sqrt{\frac{1 - \varepsilon_1}{\varepsilon_1} + \sqrt{\frac{1 - \varepsilon_1}{\varepsilon_1}} \sqrt{\frac{\varepsilon_3}{1 - \varepsilon_3}}} \quad \text{and} \quad 
U_2 = - \frac{2\varepsilon_2 + 4\varepsilon_1 \varepsilon_3}{\sqrt{4\varepsilon_1 \varepsilon_3(\varepsilon_2 + 4\varepsilon_1 \varepsilon_3)}} \sqrt{\frac{\varepsilon_3}{1 - \varepsilon_3}}.
$$
Proposition~\ref{prop:dftu3}.2 then states that, when $s_1^* \in [\mu_1, \mu_2]$, the first-step SFR fails if and only if 
\[
\mu_1 \geq F(\varepsilon_1).
\]

Recall that the case $s_1^* \in [\mu_2,\mu_3]$ need not be analyzed separately. 
Indeed, by the affine invariance of the SFR, applying the transformation $X \mapsto -X$ maps this case to $s_1^* \in [\mu_1,\mu_2]$ while leaving the SFR unchanged. 
Therefore, it suffices to consider the case $s_1^* \in [\mu_1,\mu_2]$ without loss of generality.

To study the case 
\(\mu_1 \in [F(\varepsilon_1), U(\varepsilon_1, \varepsilon_3)]\), 
we rely on numerical optimization using SciPy’s SLSQP solver \citep{2020SciPy-NMeth}. The optimization variables are $\mu_1 \in [-10, 0]$, $\varepsilon_1 \in [0.01, 0.99]$, and $\varepsilon_3 \in [0.01, 0.99]$, from which $\varepsilon_2$, $\mu_2$, and $\mu_3$ can be deduced, as detailed in the proof of Proposition~\ref{prop:dftu3}.
We impose the following constraints:
\begin{itemize}
    \item $1 - \varepsilon_1 - \varepsilon_3 > 0$,
    \item $F(\varepsilon_1) \leq \mu_1 \leq U(\varepsilon_1, \varepsilon_3)$.
\end{itemize}
 
The first two optimization problems
\[
\min \Bigl[\gamma - \mu_1 - \mu_2 \Bigr],
\quad \text{and} \quad 
\min \Bigl[\mu_2 + \mu_3 - \gamma \Bigr] 
\]
with $\gamma = \E[X^3]$, both yield positive values when solved numerically. 
This implies that 
$$s^{**} = \frac{\gamma}{2} \in \Bigl(\frac{\mu_1+\mu_2}{2}, \frac{\mu_2+\mu_3}{2}\Bigr).$$

Let $\tilde{\mu}_1$, $\tilde{\mu}_2$, and $\tilde{\mu}_3$ denote the locations after the transformation $X \mapsto |X - s^{**}|$.  
The preceding result implies that $\tilde{\mu}_1 > \tilde{\mu}_2$ and $\tilde{\mu}_3 > \tilde{\mu}_2$, so $\tilde{\mu}_2$ is the smallest of the transformed locations.

\begin{tikzpicture}[>=Stealth]

  \coordinate (mu1) at (0,0);
  \coordinate (mu2) at (3,0);
  \coordinate (mu3) at (5,0);

  \draw[thick] ($(mu1)+(-1,0)$) -- ($(mu3)+(1,0)$);

  \foreach \p/\lab in {mu1/$\mu_1$, mu2/$\mu_2$, mu3/$\mu_3$} {
    \draw[fill=black] (\p) circle (1.5pt);
    \node[below=3pt] at (\p) {\lab};
  }

  \coordinate (m12) at ($(mu1)!0.5!(mu2)$);
  \coordinate (m23) at ($(mu2)!0.5!(mu3)$);

  \draw[pattern=north east lines, pattern color=gray] 
        (m12) rectangle ($(m23)+(0,0.15)$); 
  \node[above=2pt] at ($(m12)!0.5!(mu2)$) {$s^{**}$};

  \coordinate (target_interval) at ($(mu2)!0.8!(mu3)$);
  \draw[->, thick, bend left=30, blue] (mu1) to (target_interval);

  \coordinate (target_gt_mu3) at ($(mu3)+(1,0)$);
  \draw[->, thick, bend left=35, blue] (mu1) to (target_gt_mu3);

  \node[above=25pt, blue] at ($(mu1)!0.5!(mu3)$) {$|\mu_1 - s^{**}|+s^{**}$};

\end{tikzpicture}

To apply Proposition~\ref{prop:dftu3}.2 in the second step of the DFTU, we standardize and order the locations $\tilde{\mu}_1$, $\tilde{\mu}_2$, and $\tilde{\mu}_3$. 
The resulting locations are denoted by $\mu^{(2)}_{(1)}, \mu^{(2)}_{(2)}$, and $\mu^{(2)}_{(3)}$, with
$\mu^{(2)}_{(1)} < \mu^{(2)}_{(2)} < \mu^{(2)}_{(3)}$, and corresponding proportions  $\varepsilon^{(2)}_{(1)}, \varepsilon^{(2)}_{(2)}$, and $\varepsilon^{(2)}_{(3)}$. 

To apply Proposition~\ref{prop:dftu3}.2, we need to verify that $s^*_2 \in [\mu^{(2)}_{(1)}, \mu^{(2)}_{(2)}]$, which, according to Proposition~\ref{prop:dftu3}.1, is equivalent to $\mu^{(2)}_{(1)} \leq U\big(\varepsilon^{(2)}_{(1)}, \varepsilon^{(2)}_{(3)}\big)$. This leads to the third optimization problem
\[
\min \, \Bigl[ U\big(\varepsilon^{(2)}_{(1)}, \varepsilon^{(2)}_{(3)}\big) - \mu^{(2)}_{(1)} \Bigr]
\]
whose numerical solution is positive, thereby ensuring that $s^*_2 \in [\mu^{(2)}_{(1)}, \mu^{(2)}_{(2)}]$. 

The final optimization problem
\[
\min \Bigl[ \mu^{(2)}_{(1)} - F(\varepsilon^{(2)}_{(1)}) \Bigr]
\]
attains a positive value using a numerical solver, ensuring that $ \mu^{(2)}_{(1)} < F(\varepsilon^{(2)}_{(1)})$. According to Proposition~\ref{prop:dftu3}.2, this implies that the second-step SFR cannot fail when the first-step one has failed.

\subsection{DFTU algorithm}\label{app:algo}

See Algorithm~\ref{alg:dftu}.

\begin{algorithm}[ht!]
\caption{Algorithm of the DFTU}\label{alg:dftu}

\begin{algorithmic}

\Require Data matrix $X$, critical values $q_1$ and $q_2$ 
\Ensure Unimodal or Multimodal

\State $s^*_1 \gets \argmin_{s \in \mathbb{R}} \Var|X-s|$ \Comment{Exact pivot on $X$}
\State $\Phi^*_1(X) \gets 4 \, \Var|X - s^*_1| / \Var X$
\If{$\Phi^*_1(X) < q_1$}
\State \Return Multimodal

\Else
\State $s^{**}_1 \gets \argmin_{s \in \mathbb{R}} \Var (X-s)^2$ \Comment{Approximated pivot on $X$}
\State $\tilde{X} \gets |X - s^{**}_1|$ \Comment{Folding}
\State $ s^*_2 \gets \argmin_{s \in \mathbb{R}} \Var|\tilde{X}-s|$ \Comment{Exact pivot on $\tilde{X}$}
\State $\Phi^*_2(X) \gets 4 \, \Var|\tilde{X} - s^*_2| / \Var \tilde{X}$

\If{$\Phi^*_2(X) < q_2$}
    \State \Return Multimodal
    \Else
    \State \Return Unimodal
\EndIf

\EndIf

\end{algorithmic}
\end{algorithm}

\section{Details about the simulation study}\label{app:simu}

In this section, we provide details about the simulation study. The objective is to compare the results of different unimodality tests across various mixture models. All computations were done in Python 3.13. Replication files are available from: \url{https://github.com/cbecquart/DFTU-Replication}.

For each mixture with $k$ components, we generate 100 datasets, each containing 1,000 points. Specifically, we draw $1,000\; \varepsilon_i$ points from the distribution of the $i$-th component of the mixture model, for $i \in \{1, \ldots, k\}$.  For each dataset, we apply the following tests at a 5\% significance level: 

\begin{itemize}
    \item the FTU with the exact pivot, denoted FTU$^{*}$, partially based on an implementation provided by A.~Siffer (personal communication);
    \item the FTU with the approximate pivot, denoted FTU$^{**}$, implemented as a Python translation of the R code available in the \texttt{Rfolding} package \citep{siffer_rfolding_2018};
    \item the DFTU, implemented according to Algorithm~\ref{alg:dftu}, with parameter $\alpha_1$ set to 3\% and $\alpha_2$ deduced from $\alpha$ and $\alpha_1$;
    \item Hartigan’s dip test, implementation from the \texttt{diptest} library \citep{urlus_diptest_2022}. 
\end{itemize}

The critical values for all three folding-based tests are obtained via Monte Carlo simulation, as described in Section~3.1 of the main paper. For each test, we record the number of times unimodality is concluded across the 100 datasets. Table~\ref{tab:simuraw} reports these counts. Given the strong polarization of the results, values close to 100 are reported as ``Uni" (for unimodal) and those close to 0 as ``Multi" (for multimodal) in Table~1 of the main paper. Table \ref{tab:distrib} provides details about the distributions used in the simulation study. See also the comments in the main paper.

\begin{table}[ht]
\centering
\begin{tabular}{lcccc}
\hline
Distribution & FTU$^{*}$  & FTU$^{**}$ & DFTU & dip test \\
\hline
$\mathcal{N}(0, 1)$ & 100         & 100      & 100  & 100     \\
$0.5 \cdot (\mathcal{N}_0 + \mathcal{N}_1)$ & 100         & 100      & 100  & 100     \\
$0.6 \cdot \mathcal{N}_0 + 0.4 \cdot \mathcal{U}([4, 8])$ & 0           & 0        & 0    & 0       \\
$0.6 \cdot \mathcal{N}_0 + 0.4 \cdot \mathcal{U}([1, 4])$ & 1           & 1        & 1    & 100       \\

$1/3 \cdot \delta_{-2} + 1/3 \cdot \delta_{0} + 1/3 \cdot \delta_{2}$ & 100         & 100      & 0    & 0       \\
$1/3 \cdot \mathcal{N}_{-2} + 1/3 \cdot \mathcal{N}_{0} + 1/3 \cdot \mathcal{N}_{2}$ & 100         & 100      & 4    & 0       \\
$0.2 \cdot (\delta_{-3} + \delta_{-1.5} + \delta_{2.5} + \delta_{4} + \delta_{11})$ & 100         & 100      & 0    & 0       \\
$0.2 \cdot (\mathcal{N}_{-3} + \mathcal{N}_{-1.5} + \mathcal{N}_{2.5} + \mathcal{N}_{4} + \mathcal{N}_{11})$ & 100         & 100      & 0    & 0       \\
\hline
\end{tabular}
\caption{Simulation results. Counts of unimodality across the 100 datasets for the FTU with exact pivot (FTU$^{*}$), the FTU with approximate pivot (FTU$^{**}$), the DFTU, and the dip test. $\mathcal{N}_a = \mathcal{N}(a, 0.5)$.}
\label{tab:simuraw}
\end{table}

\begin{landscape}
\begin{table}
\centering

\begin{tabular}{T{6cm} C{3cm} T{3cm} T{4cm}}
\toprule
 Distribution & Histogram
 & ECDF
 &  Description \\
 \hline
 $\mathcal{N}(0, 1)$ & \includegraphics[width=\linewidth,height=3cm,keepaspectratio]{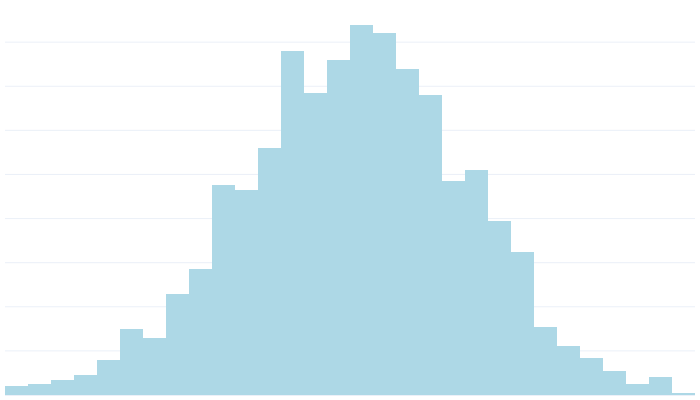}
 & \includegraphics[width=\linewidth,height=3cm,keepaspectratio]{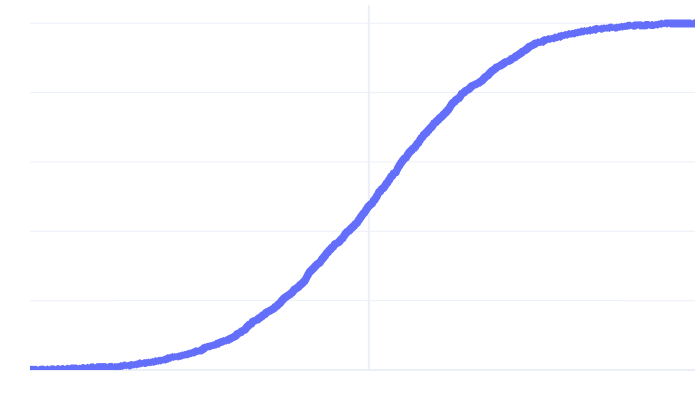}
 & Simple case of unimodality \\ 
 $0.5 \cdot (\mathcal{N}_0 + \mathcal{N}_1)$ & \includegraphics[width=\linewidth,height=3cm,keepaspectratio]{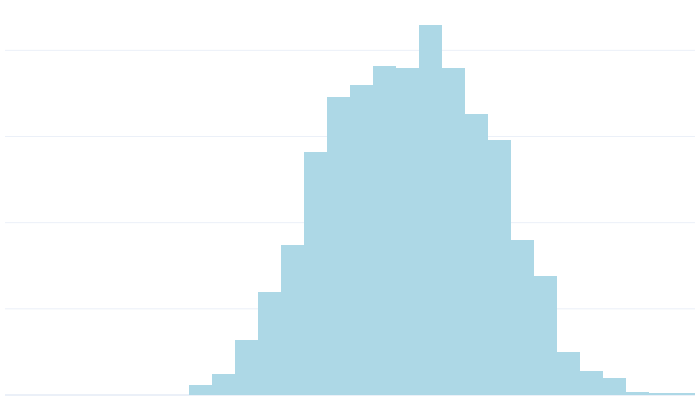}
 & \includegraphics[width=\linewidth,height=3cm,keepaspectratio]{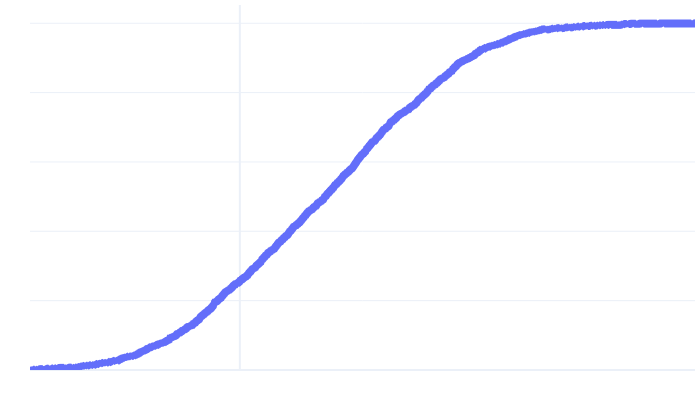}
 & 2-Gaussian mixture, not well separated  \\ 
 $0.6 \cdot \mathcal{N}_0 + 0.4 \cdot \mathcal{U}([4, 8])$ & \includegraphics[width=\linewidth,height=3cm,keepaspectratio]{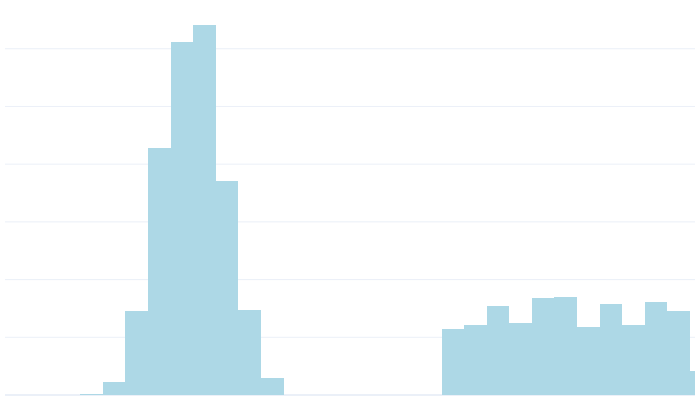}
 & \includegraphics[width=\linewidth,height=3cm,keepaspectratio]{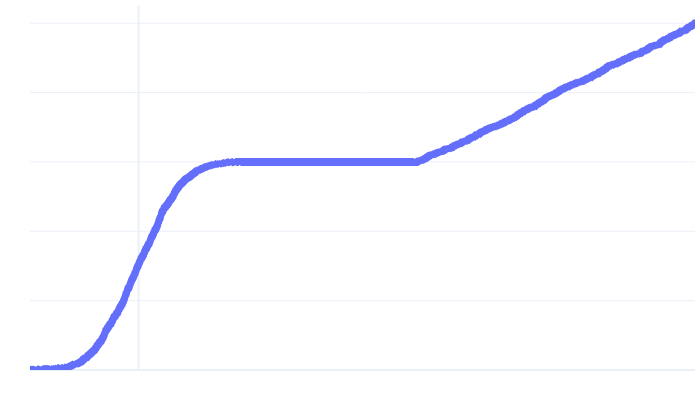}
 & Mixture of Gaussian and Uniform distributions, well separated \\ 
 $0.6 \cdot \mathcal{N}_0 + 0.4 \cdot \mathcal{U}([1, 4])$ & \includegraphics[width=\linewidth,height=3cm,keepaspectratio]{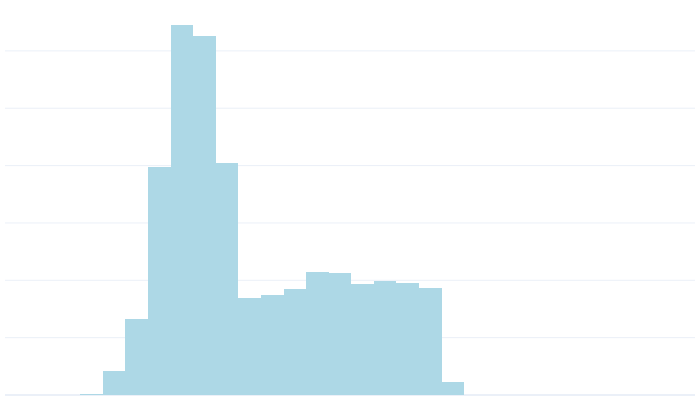}
 & \includegraphics[width=\linewidth,height=3cm,keepaspectratio]{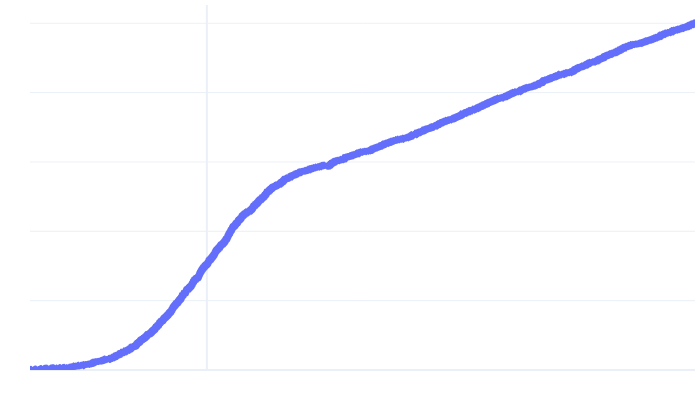}
 & Mixture of Gaussian and Uniform distributions, less separated \\
 $0.2 \cdot \delta_{-2} + 0.4 \cdot \delta_{0} + 0.4 \cdot \delta_{2}$ & \includegraphics[width=\linewidth,height=3cm,keepaspectratio]{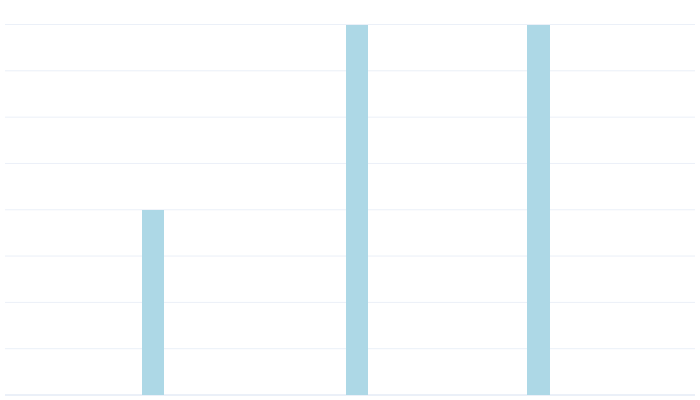}
 & \includegraphics[width=\linewidth,height=3cm,keepaspectratio]{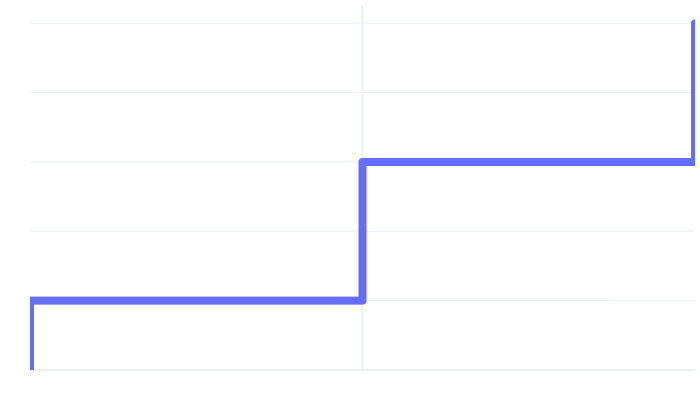}
 & 3-Dirac mixture \\
 $0.2 \cdot \mathcal{N}_{-2} + 0.4 \cdot \mathcal{N}_{0} + 0.4 \cdot \mathcal{N}_{2}$ & \includegraphics[width=\linewidth,height=3cm,keepaspectratio]{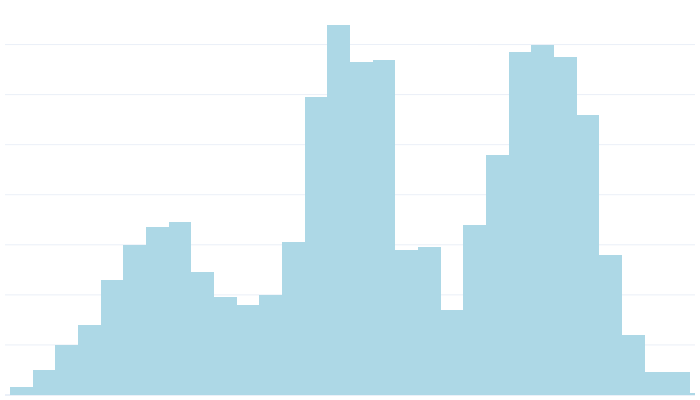}
 & \includegraphics[width=\linewidth,height=3cm,keepaspectratio]{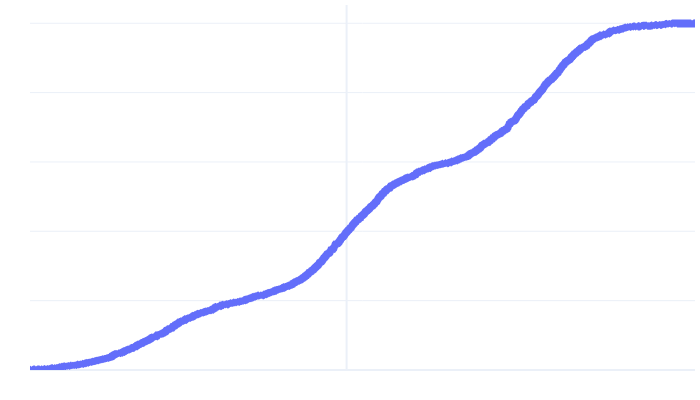}
 & 3-Gaussian mixture, well separated \\

\bottomrule
\end{tabular}

\caption{Details about the distributions used in the simulations.}
\label{tab:distrib}

\end{table}
\end{landscape}

\begin{landscape}
\begin{table}\ContinuedFloat
\centering

\begin{tabular}{T{6cm} C{3cm} T{3cm} T{4cm}}

\toprule
 Distribution & Histogram
 & ECDF
 &  Description \\
\hline
 $0.2 \cdot (\delta_{-3} + \delta_{-1.5} + \delta_{2.5} + \delta_{4} + \delta_{11})$ & \includegraphics[width=\linewidth,height=3cm,keepaspectratio]{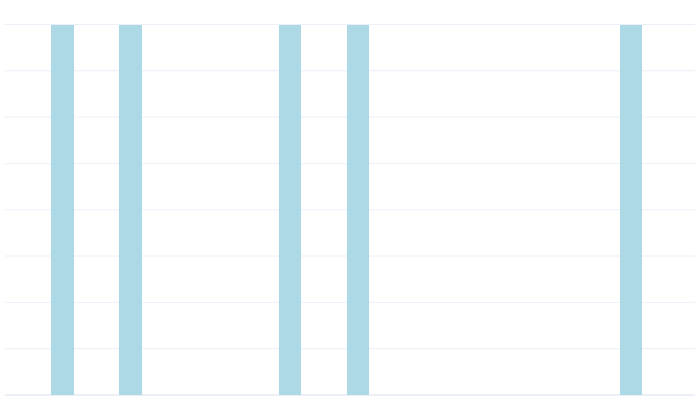}
 & \includegraphics[width=\linewidth,height=3cm,keepaspectratio]{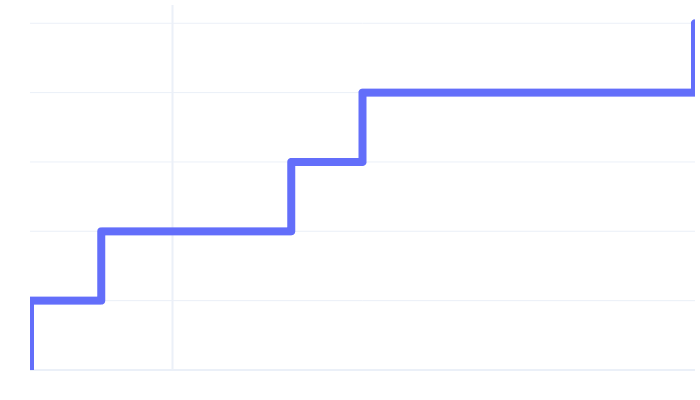}
 & 5-Dirac mixture \\
 $0.2 \cdot (\mathcal{N}_{-3} + \mathcal{N}_{-1.5} + \mathcal{N}_{2.5} + \mathcal{N}_{4} + \mathcal{N}_{11})$ & \includegraphics[width=\linewidth,height=3cm,keepaspectratio]{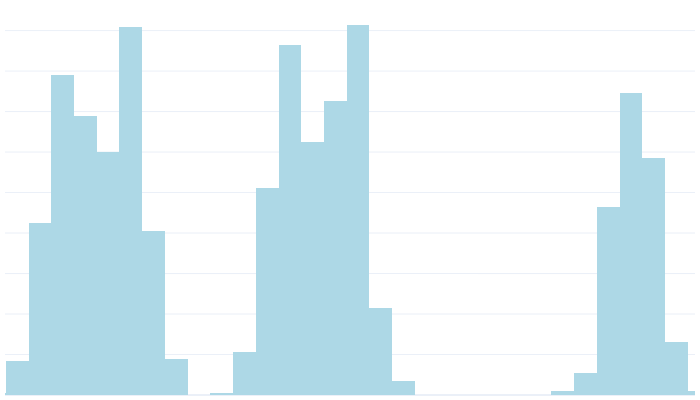}
 & \includegraphics[width=\linewidth,height=3cm,keepaspectratio]{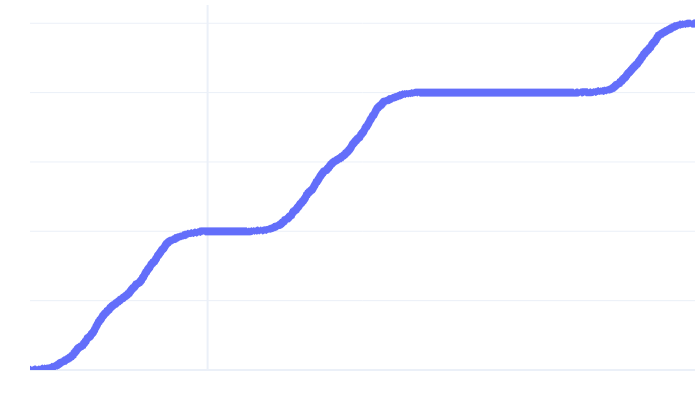}
 & 5-Gaussian mixture, well separated \\
\bottomrule
\end{tabular}

\caption{Details about the distributions used in the simulations (Continued).}
\end{table}
\end{landscape}

\end{appendix}

\section*{Acknowledgments}
We thank Alban Siffer for stimulating discussions and code about the FTU. 
A.M.R. and C.B. acknowledge funding from the French National Research Agency (ANR) under the Investments for the Future (Investissements d’Avenir) program, grant ANR-17-EURE-0010. Part of C.B.'s work was done during an Erasmus+ visit funded by the University of Toulouse at the University of Helsinki.

\bibliographystyle{unsrtnat}
\bibliography{references_v1}

\end{document}